\documentclass[11pt]{article}
\usepackage{subfigure}
\usepackage{color}
\usepackage{graphicx,amssymb,amsmath}
\usepackage{citesort}

\usepackage{geometry}
\usepackage{fullpage}
\usepackage[final]{hyperref} 
\usepackage{float}

\usepackage{appendix}
\usepackage{graphicx}
\usepackage{amsfonts}
\usepackage{amssymb}
\usepackage{amsmath}
\usepackage{amsthm}
\usepackage{url}
\usepackage{cite}

\hypersetup{
	colorlinks=true,       
	linkcolor=blue,          
	citecolor=blue,        
	filecolor=magenta,      
	urlcolor=blue
}

\newtheorem{lemma}{Lemma}
\newtheorem{theorem}[lemma]{Theorem}
\newtheorem{definition}[lemma]{Definition}

\newtheorem{corollary}[lemma]{Corollary}
\newtheorem{proposition}[lemma]{Proposition}

\newenvironment{proofoflm}[1]{\noindent {\em Proof of Lemma #1: }\ignorespaces}{}

\newcommand{\topic}[1]{\vspace{0.2cm}\noindent{\bf #1:}}

\newcommand{\calA}{\mathcal{A}}
\newcommand{\calG}{\mathcal{G}}
\newcommand{\calH}{\mathcal{H}}
\newcommand{\calS}{\mathcal{S}}
\newcommand{\calP}{\mathcal{P}}

\newcommand{\R}{\mathbb{R}}

\newcommand{\ConstC}{ C }

\newcommand{\OPT}{ \mathsf{OPT}}
\newcommand{\SOL}{ \mathsf{SOL}}

\newcommand{\setcover}{{\sf SC}}

\newcommand{\setDisk}{\mathcal{D} }
\newcommand{\Disk}{\mathsf{D} }
\newcommand{\dcenter}{D}
\newcommand{\point}{P}
\newcommand{\setssquare}{\Xi}
\newcommand{\ssquare}{\Gamma}
\newcommand{\cycle}{\mathsf{Cyc}}

\newcommand{\setPoint}{\mathcal{P} }

\newcommand{\baseline}{\mathsf{b}}
\newcommand{\substructure}{\mathsf{St}}
\newcommand{\activeR}{\mathsf{Ar}}

\newcommand{\UDC}{\textsf{WUDC}}
\newcommand{\MWCDS}{\textsf{MWCDS}}
\newcommand{\MWDS}{\textsf{MWDS}}
\newcommand{\NWST}{\textsf{NWST}}
\newcommand{\MLC}{\textsf{MLC}}
\newcommand{\MSC}{\textsf{MSC}}

\newcommand{\gadget}{\mathsf{Gg}}
\newcommand{\auxgraph}{\mathfrak{S}}
\newcommand{\diskarcgraph}{\mathfrak{B}}
\newcommand{\planargraph}{\mathfrak{D}}

\newcommand{\eat}[1]{}
\newcommand{\Center}{\mathfrak{C}}
\newcommand{\coreCenter}{\mathfrak{C}_{o}}
\newcommand{\apex}{ \theta_{\coreCenter} }
\newcommand{\Path}{ \mathsf{Path} }

\newcommand{\Region} {\mathbb{R}}
\newcommand{\Reg}{\mathbb{R}}

\newcommand{\Uncover}{\mathbb{U}(\calH)}

\newcommand{\toparcP}{{t_P}}
\newcommand{\basearcP}{{b_P}}

\newcommand{\basearcQ}{{b_Q}}

\newcommand{\calPco}{\mathcal{P}_{\text{co}}}

\newcommand{\ReadyDisk}{\mathcal{R} }
\newcommand{\Neighbor}{N_{\text{p}}(\Disk) }
\newcommand{\state}{\Phi}
\newcommand{\Pathset}{\mathsf{PS}}

\newcommand{\match}{\pi}
\newcommand{\event}{\mathcal{E}}
\newcommand{\poly}{\mathrm{poly}}

\newcommand{\true}{\textrm{True}}

\newcommand{\kset}[1]{ \{#1 \}_{k\in[m]}}

\newcommand{\Domain}{ \mathsf{Dom} }

\title{A PTAS for the Weighted Unit Disk Cover Problem}
\author{Jian Li \footnote{Email:lijian83@mail.tsinghua.edu.cn}
\quad\quad\quad\quad\quad\quad
Yifei Jin \footnote{Email: jin-yf13@mails.tsinghua.edu.cn} \\
Institute for Interdisciplinary Information Sciences\\
Tsinghua University, China
}

\begin{document}




\maketitle

\begin{abstract}
We are given a set of weighted unit disks and a set of points in Euclidean plane.  The minimum
weight unit disk cover (\UDC) problem asks for a subset of disks of minimum total weight that covers
all given points.  \UDC\ is one of the geometric set cover problems, which have been
studied extensively for the past two decades (for many different geometric range spaces, such as
(unit) disks, halfspaces, 
rectangles, triangles).  
It is known that the unweighted \UDC\ problem
is NP-hard and admits a polynomial-time approximation scheme (PTAS).  
For the weighted
\UDC\ problem, several constant approximations have been developed. 
However, whether the problem admits a PTAS has been an open question.
In this paper, we answer this question affirmatively 
by presenting the first PTAS for \UDC.
Our result implies the first PTAS for the minimum weight dominating set problem 
in unit disk graphs.
Combining with existing ideas, our result can also be used to obtain
the first PTAS for the maximum lifetime coverage problem and an improved 
4.475-approximation
for the connected dominating set problem in unit disk graphs.
\end{abstract}

\newpage

\section{Introduction}
\label{sec:intro}

The set cover problem is a central problem in theoretical computer science 
and combinatorial optimization.
A set cover consists of a ground set $U$ and collection $\mathcal{S}$ of subsets of $U$. Each set $S\in \mathcal{S}$ has a non-negative weight $w_S$. The goal 
is to find a subcollection $\mathcal{C} \subseteq \mathcal{S}$ of minimum total weight 
such that $\bigcup\mathcal{C}$ covers all elements of $U$. 
The approximibility of the general \setcover\ problem is rather well understood:
it is well known that the greedy algorithm is an $H_n$-approximation ($H_n=\sum_{i=1}^n 1/i$)
and obtaining a $(1-\epsilon)\ln n$-approximation for any constant $\epsilon>0$ 
is NP-hard \cite{feige1998threshold,dinur2014analytical}.
In the \emph{geometric set cover problem},
$U$ is a set of points in some Euclidean space $\R^d$,
and $\calS$ consists of geometric objects (e.g., disks, squares, triangles).
In the geometric setting, we can hope for better-than-logarithmic approximations
due to the special structure of $\calS$. 
Most geometric
set cover problems are still NP-hard, even for the very simple classes of objects such as unit
disks \cite{Hochbaum1987305, Clark91} (see \cite{Chan2014112,Peled12} for more examples
and exceptions).
Approximation algorithms for geometric set cover
have been studied extensively for the past two decades,
not only because of the importance of the problem per se,
but also its rich connections to other important notions and problems, such as
VC-dimension~\cite{Bronnimann, Clarkson, even2005hitting}, $\epsilon$-net, union
complexity~\cite{Varadarajan2009, varadarajan2010weighted, Chan2012},  
planar separators~\cite{mustafa2009ptas, gibson2010algorithms}, even machine scheduling
problems~\cite{bansal2014geometry}.  

In this work, we study the geometric set cover problem with one of the 
simplest class of objects $-$ unit disks.
The formal definition of our problem is as follows:
\begin{definition}
Weighted Unit Disk Cover (\UDC):
  Given a set $\setDisk=\{\Disk_1,\ldots,\Disk_n\}$ of $n$ unit disks
  and a set $\setPoint=\{\point_1,\ldots,\point_m\}$ of $m$ points in Euclidean plane $\mathbb{R}^2$.
  Each disk $\Disk_i$ has a weight $w(\Disk_i)$. 
  Our goal is to choose a subset of disks to cover all points in
  $\setPoint$, and the total weight of the chosen disks is minimized.
\end{definition}

\UDC\ is the general version of  the following minimum weight dominating set problem in unit disk graphs (UDG).

\begin{definition}
\label{df:DS}
Minimum Weight Dominating Set (\MWDS) in UDG:
We are given a unit disk graph $G(V,E)$, where $V$ is a set of weighted points in $\R^2$
and $(u,v)\in E$ iff $\|u-v\|\leq 1$ for any $u,v\in V$. 
A dominating set $S$ is a subset of $V$ such that 
for any node $v\not\in S$, there is some $u\in S$ with $(u,v)\in E$.
The goal of the minimum weight dominating set problem
is to find a dominating set with the minimum total weight.
\end{definition}
	
Given a dominating set instance with point set $V$, we create a \UDC\ instance by placing, 
for each point in $v\in V$, a point (to be covered) co-located with $v$ and a unit disk centered at
$v$, with weight equal to the weight of $v$.  Thus, \UDC\ is a general version of problem \MWDS\ in
UDG. 
In this paper, we only state our algorithms and results in the context of \UDC. 

\subsection{Previous Results and Our Contribution}
\label{sec:ourresult}

We first recall that a polynomial time approximation scheme (PTAS) for a minimization problem is an algorithm $\calA$
that takes an input instance, a constant $\epsilon > 0$, returns a solution $\SOL$ such that
$\SOL \leq (1+\epsilon)\OPT$, where $\OPT$ is the optimal value,
and the running time of $\calA$ is polynomially in the size of the input for any fixed constant $\epsilon$.

\UDC\ is NP-hard, even for the unweighted version 
(i.e., $w(\Disk_i)=1$) \cite{Clark91}.
For unweighted dominating set in unit disk graphs, 
Hunt et al.~\cite{Hunt9} obtained the first PTAS
in unit disk graphs. 
For the more general disk graphs, based on the connection between geometric set cover problem and $\epsilon$-nets, 
developed in \cite{Bronnimann, even2005hitting,  Clarkson}, and the existence of $\epsilon$-net of
size $O(1/\epsilon)$ for halfspaces in $\R^3$~\cite{pyrga2008new} (see also \cite{har2014epsilon}), 
it is possible to achieve a constant factor approximation. As estimated in \cite{mustafa2009ptas},
these constants are at best 20 (A recent result~\cite{Norbert15} shows that the constant is at most
13). Moreover, there exists a PTAS for unweighted disk 
cover and minimum dominating set   
via the local search technique~\cite{mustafa2009ptas, gibson2010algorithms}.

For the general weighted \UDC\ problem, the story is longer. Amb\"{u}hl et al.~\cite{Ambuhl06}
obtained the first approximation for \UDC\ 
with a concrete constant 72.
Applying the shifting technique of \cite{hochbaum1985approximation},
Huang et al.~\cite{Huang09} obtained a $(6+\epsilon)$-approximation algorithm for \UDC. 
The approximation factor was later improved to $(5+\epsilon)$ \cite{dai20095},
and to $(4+\epsilon)$ by several groups  \cite{Erlebach10, zou2011new, ding2012constant}.
The current best ratio is 3.63. 
\footnote{
	The algorithm can be found in Du and Wan
	\cite{du2012connected}, who attributed the result to
	a manuscript by Willson et al. 
	} 
Besides, the \emph{quasi-uniform sampling method}~\cite{varadarajan2010weighted, Chan2012} provides
another approach to  achieve a constant 
factor approximation for \UDC\ (even in disk graphs). 
However, the constant  depends on several other constants from rounding LPs and the size of 
\emph{the union complexity}.  Very recently, based on
the separator framework of Adamaszek and 
Wiese~\cite{adamaszek2013approximation}, Mustafa and Raman~\cite{mustafa2014qptas} 
obtained a QPTAS
(Quasi-polynomial time approximation scheme)
 for weighted disks in $\R^2$ (in fact, weighted halfspaces in $\R^3$), thus ruling
out the APX-hardness of \UDC. 

Another closely related work is by Erlebach and van Leeuwen~\cite{ErlebachSquare},
who obtained a PTAS for set cover on weighted unit squares, which is the first
PTAS for weighted geometric set cover on any planar objects (except those poly-time solvable cases~\cite{Chan2014112,Peled12}).
Although it may seem that their result is quite close to a PTAS for weighted \UDC,
as admitted in their paper, their technique is insufficient for handling
unit disks and ``completely different insight is required''.

In light of all the aforementioned results, 
it seems that we should expect a PTAS for \UDC, but it 
remains to be an open question 
(explicitly mentioned as an open problem in a number of 
previous papers, e.g., 
\cite{Ambuhl06, Erlebach10, ErlebachSquare, erlebach2011maximising, du2012connected, van2009optimization}
). 
Our main contribution in this paper is to settle this question affirmatively
by presenting the first PTAS for \UDC.

\begin{theorem}
	\label{thm:UDC}
	There is a polynomial time approximation scheme for the \UDC\ problem.
	The running time is 
    $n^{O(1/\epsilon^9)}$.
\end{theorem}

Because \UDC\  is more general than \MWDS, 
we immediately have the following corollary.

\begin{corollary}
	\label{cor:DS}
	There is a polynomial time approximation scheme for the minimum weight dominating set problem
    in unit disk graphs. 
\end{corollary}

We note that the running time $n^{\poly(1/\epsilon)}$ is nearly optimal in light of the negative result by Marx~\cite{marx2007optimality}, who showed that an EPTAS
(i.e., 
Efficient PTAS,  with running time $f(1/\epsilon)\poly(n)$
) even for the unweighted 
dominating set in UDG would contradict the exponential time hypothesis.

Finally, in Section~\ref{sec:app},
we show that our PTAS for \UDC\ can be used to obtain improved approximation algorithms
for two important problems in wireless sensor networks,
the connected dominating set problem and the maximum lifetime coverage problem in UDG.

\section{Our Approach - A High Level Overview}
\label{sec:overview}

By the standard shifting technique\cite{Du11}, it suffices to provide a PTAS for \UDC\ when
all disks lies
in a square of constant size (we call it a block, and the constant depends on $1/\epsilon$).
This idea is formalized in Huang et al.~\cite{Huang09}, as follows.

\begin{lemma}[Huang ta al.~\cite{Huang09}]
	\label{lemma:02:01}
	Suppose there exists a $\rho$-approximation for \UDC\ in a fixed $L\times L$ block,
	with running time $f(n, L)$. Then there exists a
	$(\rho+O(1/L))$-approximation with running time $O(L\cdot n \cdot f(n, L))$ for \UDC.
	In particular, setting $L=1/\epsilon$,
	there exists a $(\rho+\epsilon)$-approximation for \UDC,
	with running time $O\left(\frac{1}{\epsilon}\cdot n\cdot f(n, \frac{1}{\epsilon})\right)$.
\end{lemma}

In fact, almost all previous constant factor approximation algorithms for \UDC\
were obtained by developing constant approximations for a single block of a constant size
(which is the main difficulty
\footnote{
	For the unweighted dominating set problem in a single block, it is easy to see the optimal number of disks
	is bounded by a constant, which implies that we can compute the optimum in poly-time.
	However, for the weighted dominating set problem or \UDC\, the optimal solution in a single
    block may consist of $O(n)$ disks.
}).
The main contribution of the paper is to improve on the previous work~\cite{Ambuhl06,Huang09,dai20095,Erlebach10} for a single block,
as in the following lemma.

\begin{lemma}
	\label{thm:02}
	There exists a PTAS for \UDC\ in a fixed block of size
	$L\times L$ for $L=1/\epsilon$.
	The running time of the PTAS is $n^{O(1/\epsilon^9)}$
\end{lemma}

From now on, the approximation error guarantee $\epsilon>0$ is a fixed constant.
Whenever we say a quantity is a constant, the constant may depend on $\epsilon$.
We use $\OPT$ to represent the optimal solution (and the optimal value) in this block.
We use capital letters $A,B,C,\ldots$ to denote points, and small letters $a,b,c,\ldots$
to denote arcs.
For two points $A$ and $B$, we use $|AB|$ to denote the line segment connecting $A$ and $B$
(and its length).
We use $\Disk_i$ to denote a disk and $\dcenter_i$ to denote its center.
For a point $A$ and a real $r>0$, let $\Disk(A,r)$ be the disk centered at $A$ with radius $r$.
For a disk $\Disk_i$, we use $\partial \Disk_i$ to denote its boundary. We call a segment of $\partial \Disk_i$ {\em an arc}.

First, we guess that whether $\OPT$ contains more than $\ConstC$ disks or not for some constant $C$.
If $\OPT$ contains no more than $\ConstC$ disks, we
enumerate all possible combinations and choose the one which covers all points and has the minimum weight.
This takes $O\left( \sum_{i=1}^{C} {n\choose i}\right) = O(n^C)$ time, which is polynomial.

The more challenging case is whether $\OPT$ contains more than $\ConstC$ disks.  In this case, we
guess (i.e., enumerate all possibilities) the set $\calG$ of the $\ConstC$ most expensive disks in
$\OPT$.  There are at most a polynomial number (i.e., $O(n^C)$) possible guesses.  Suppose our guess
is correct.  Then, we delete all disks in $\calG$ and all points that are covered by $\calG$.  Let
$\Disk_t$ (with weight $w_t$) be the cheapest disk in $\calG$.  We can see that
$\OPT \geq \ConstC
w_t$.  Moreover, we can also safely ignore all disks with weight larger than $w_t$ (assuming that
our guess is correct).  Now, our task is to cover the remaining points with the remaining disks,
each having weight at most $w_t$.  We use $\setDisk'=\setDisk\setminus \calG$ and $\setPoint' =
\setPoint\setminus \setPoint(\calG) $ to denote the set of the remaining disks and the set of
remaining points respectively, where $\setPoint(\calG)$ denote the set of points covered by some
disk in $\calG$.

Next, we carefully choose to include in our solution a set $\calH\subseteq \setDisk'$ of at most $\epsilon C$ disks.
The purpose of $\calH$ is to break the whole instance into many (still a constant)
small pieces (substructures), such that each substructure can be solved optimally, via dynamic programming.
\footnote{
An individual substructure can be solved using a dynamic program similar to \cite{Ambuhl06, Peled12}.
}
One difficulty is that the substructures are not independent and
may interact with each other (i.e., a disk may appear in more than one substructure).
In order to apply the dynamic programming technique to all substructures simultaneously,
we have to ensure the orders of the disks in different substructures
are consistent with each other.
Choosing $\calH$ to ensure a globally consistent order of disks is in fact
the main technical challenge of the paper.

Suppose we have a set $\calH$ which suits our need (i.e., the remaining instance
$(\setDisk'\setminus \calH, \setPoint'\setminus \setPoint(\calH))$ can be solved optimally in polynomial time
by dynamic programming).
Let $\calS$ be the optimal solution of the remaining instance.
Our final solution is $\SOL=\calG\cup\calH\cup\calS$.
First, we can see that
$$
w(\calS) \leq w(\OPT-\calG-\calH) \leq \OPT-w(\calG),
$$
since $\OPT-\calG-\calH$ is a feasible solution for
the instance
$(\setDisk'\setminus \calH, \setPoint'\setminus \setPoint(\calH))$.
Hence, we have that
$$
\SOL=w(\calG)+w(\calH)+w(\calS) \leq \OPT + \epsilon \ConstC w_t \leq (1+\epsilon)\OPT,
$$
where
 the 2nd to last inequality holds because $|\calH|\leq \epsilon C$, and
 the last inequality uses the fact that
$\OPT\geq w(\calG)\geq \ConstC w_t$.

\topic{Constructing $\calH$}
Now, we provide a high level sketch for how to construct $\calH\subseteq \setDisk'$.  First, we
partition the block into {\em small squares} with side length $\mu = O(\epsilon)$ such that  any disk centered
in a square can cover the whole square and the disks in the same square are close enough. Let the
set of small squares be
$\setssquare=\{\ssquare_{ij}\}_{1\leq i,j \leq K}$ where $K=L/\mu$.
For a small square $\ssquare$, let
$\Disk_{s_{\ssquare}}\in \ssquare$ and $\Disk_{t_{\ssquare}}\in \ssquare$ be the furthest pair of
disks (i.e., $|\dcenter_{s_{\ssquare}}\dcenter_{t_{\ssquare}}|$ is maximized).
We include the pair
$\Disk_{s_{\ssquare}}$ and $\Disk_{t_{\ssquare}}$ in $\calH$, for every small square $\ssquare\in
\setssquare$, and call the pair {\em the square gadget} for $\ssquare$.
See Figure~\ref{fig:gadget} for an example.
We only need to focus on covering the remaining points
in the \emph{uncovered region} $\Uncover$.

We consider all disks in a small square $\ssquare$.
The uncovered portion of those disks defines two disjoint connected regions
(See right hand side of Figure~\ref{fig:gadget}, the two shaded regions).
We call such a region, together with all relevant arcs, a {\em substructure} (formal definition in Section~\ref{sec:substructure}).	
In fact, we can solve the disk covering problem for a single substructure
optimally using dynamic programming (which is similar to the dynamic program
in\cite{Ambuhl06, Peled12}).
It appears that we are almost done, since
(``intuitively'') all square gadgets have already covered much area of the entire block,
and we should be able to use similar
dynamic program to handle all such substructures as well.
However, the situation is more complicated (than we
initially expected) since the arcs are dependent. See
Figure~\ref{fig:04:03} for a ``not-so-complicated'' example.
Firstly, there may exist two arcs (called {\em sibling arcs}) which belong to the
same disk when the disk is centered in the \emph{core-central area}, as show in
Figure~\ref{fig:gadget}).
The dynamic program has to make decisions for two sibling arcs, which belong
to two different substructures (called R(emotely)-correlated substructures), together.
Second, in order to carry out dynamic program, we need
a suitable order of all arcs. To ensure such an order exists,
we need all substructures interact with each other ``nicely".

In particular, besides all square gadgets, we need to add into $\calH$
a constant number of extra disks. This is done by a series of ``cut" operations.
A cut can either break a cycle, or break one substructure into
two substructures.
To capture how substructures interact,
we define an auxiliary graph, call substructure relation graph $\auxgraph$,
in which each substructure is a node.
The aforementioned R-correlations define a set of blue edges,
and geometrically overlapping relations define a set of red edges.
Though the cut operations, we can make blue edges form a matching, and red edges also
form a matching, and $\auxgraph$ acyclic (we call $\auxgraph$ an acyclic 2-matching).
The special structure of $\auxgraph$ allows us to define an ordering of all arcs easily.
Together with some other simple properties,
we can generalize the dynamic program from one substructure to all substructures simultaneously.

\section{Square Gadgets}
\label{sec:04}

We discuss the structure of a square gadget $\gadget(\ssquare)$ associated with the small square
$\ssquare$.  Recall that the square gadget $\gadget(\ssquare)=\Disk_s\cup\Disk_t$, where $\Disk_s$
and $\Disk_t$ are the furthest pair of disks in $\ssquare$.  We can see that for any disk $\Disk_i$
in $\ssquare$, there are either one or two arcs of $\partial\Disk_i$ which are not covered by
$\gadget(\ssquare)$.  Without loss of generality, assume that $\dcenter_s\dcenter_t$ is horizontal.
The line $\dcenter_s\dcenter_t$ divides the whole plane into two
half-planes which are denoted by  $H^{+}$ (the upper half-plane) and $H^{-}$ (the lower half-plane).
$\partial\Disk_s$ and $\partial\Disk_t$ intersect at two points $P$ and $Q$.
We need a few definitions which are useful
throughout the paper.
Figure~\ref{fig:gadget} shows an example of a square gadget.

\begin{figure}[t]
  \centering
  \includegraphics[width=0.8\textwidth]{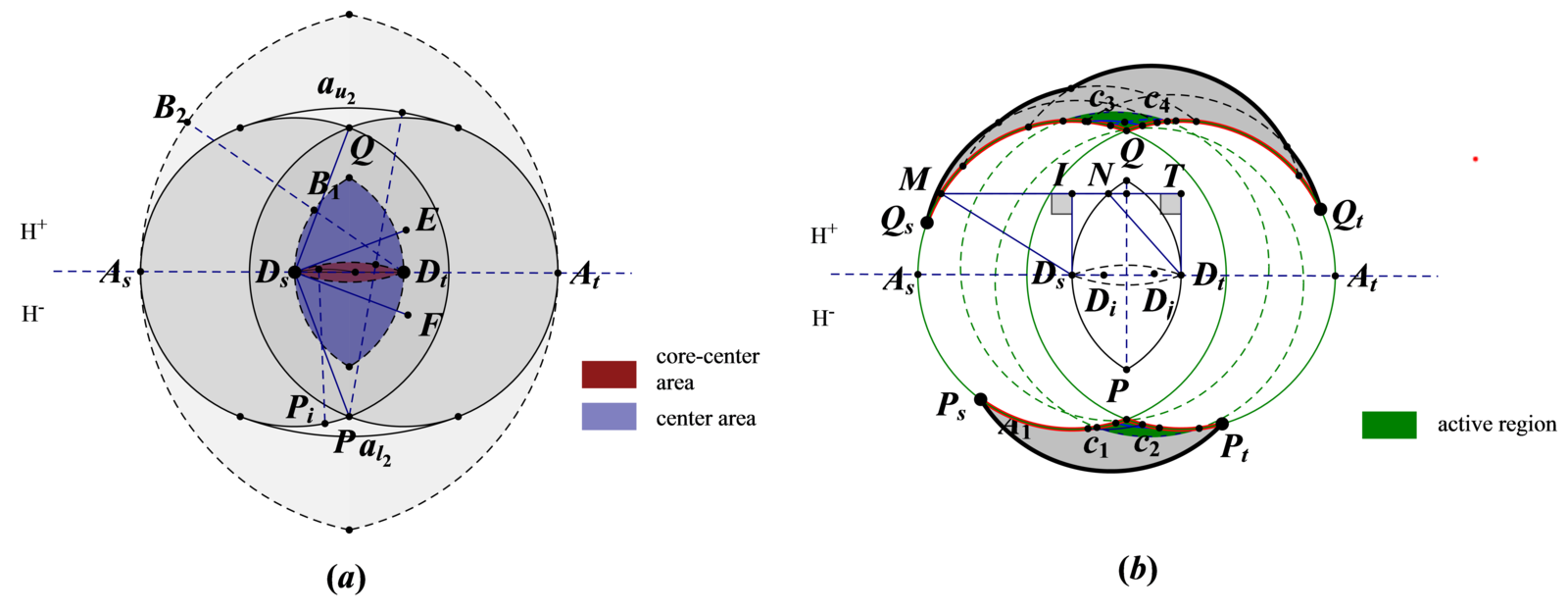}
  \caption{A square gadget. $\Disk_s$ and $\Disk_t$ are the furthest pair of disks in square
$\ssquare$ whose centers are $\dcenter_s$ and $\dcenter_t$.
On the left hand side, the blue region is the central area
$\Center=\Disk(\dcenter_s, r_{st}) \cap \Disk(\dcenter_t, r_{st})$,
where $r_{st} = |D_sD_t|$.
The brown region is the core-central area $\coreCenter=\Disk(P, 1)\cap\Disk(Q,1)$.
On the right hand side, the green area is the active regions,
defined as $\left(\bigcup_{i \in\coreCenter} \Disk_i - ( \Disk_s \cup \Disk_t)\right) \cap H^{+}$
and $\left(\bigcup_{i \in\coreCenter} \Disk_i - ( \Disk_s \cup \Disk_t)\right) \cap H^{-}$. }
  \label{fig:gadget}
\end{figure}

\begin{enumerate}
\item (Central Area and Core-Central Area)
Define the \emph{central area} of $\gadget(\ssquare)$ as the intersection of the two disks
$\Disk(\dcenter_s, r_{st})$ and $\Disk(\dcenter_t, r_{st})$ in the square $\ssquare$, where
$r_{st}=|\dcenter_s\dcenter_t|$. We use $\Center$ to denote it.
Since $\Disk_s$ and $\Disk_t$ are the furthest pair, we can
see that every other disk in $\ssquare$ is centered in the central area $\Center$.

We define the \emph{core-central area} of $\gadget(\ssquare)$ is the intersection of two unit
disks centered at $P, Q$ respectively.
Essentially, any unit disk centered in the core-central area has four
intersections with the boundary of gadget. Let us denote the area by $\coreCenter$.


\item (Active Region) Consider the regions
$$
\left( \bigcup_{\dcenter_i \in \coreCenter} \Disk_i - ( \Disk_s \cup
\Disk_t) \right) \cap H^{+}
\quad\text{ and }\quad
\left( \bigcup_{\dcenter_i \in \coreCenter} \Disk_i - ( \Disk_s \cup
\Disk_t) \right) \cap H^{-}.
$$
We call each of them an
\emph{active region} associated with square $\ssquare$.
An active region can be covered by disks centered in the core-central area.
We use $\activeR$ to denote an active region.

\end{enumerate}

\section{Substructures}
\label{sec:substructure}

Initially, $\calH$ includes all square gadgets.
In Section~\ref{sec:calH}, we will include in $\calH$ a constant number of extra disks.
For a set $S$ of disk, we use $\Reg(S)$ to denote the region covered by disks
in $S$ (i.e., $\cup_{\Disk_i\in S} \Disk_i$).
Assuming a fixed $\calH$,
we now describe the basic structure of the uncovered region
$\Reg(\setDisk')-\Reg(\calH)$.
\footnote{
Recall that $\setDisk'=\setDisk\setminus \calG$ where $\calG$ is the $\ConstC$ most expensive disks in $\OPT$.
}
For ease of notation, we use $\Uncover$ to denote the uncovered region $\Reg(\setDisk')-\Reg(\calH)$.
Figure~\ref{fig:04:03} shows an example. Intuitively, the
region consists of several ``strips'' along the boundary of $\calH$.
Now, we define some notions to describe the structure of those strips.

\begin{figure}[t]
  \centering
  \includegraphics[width=0.5\textwidth]{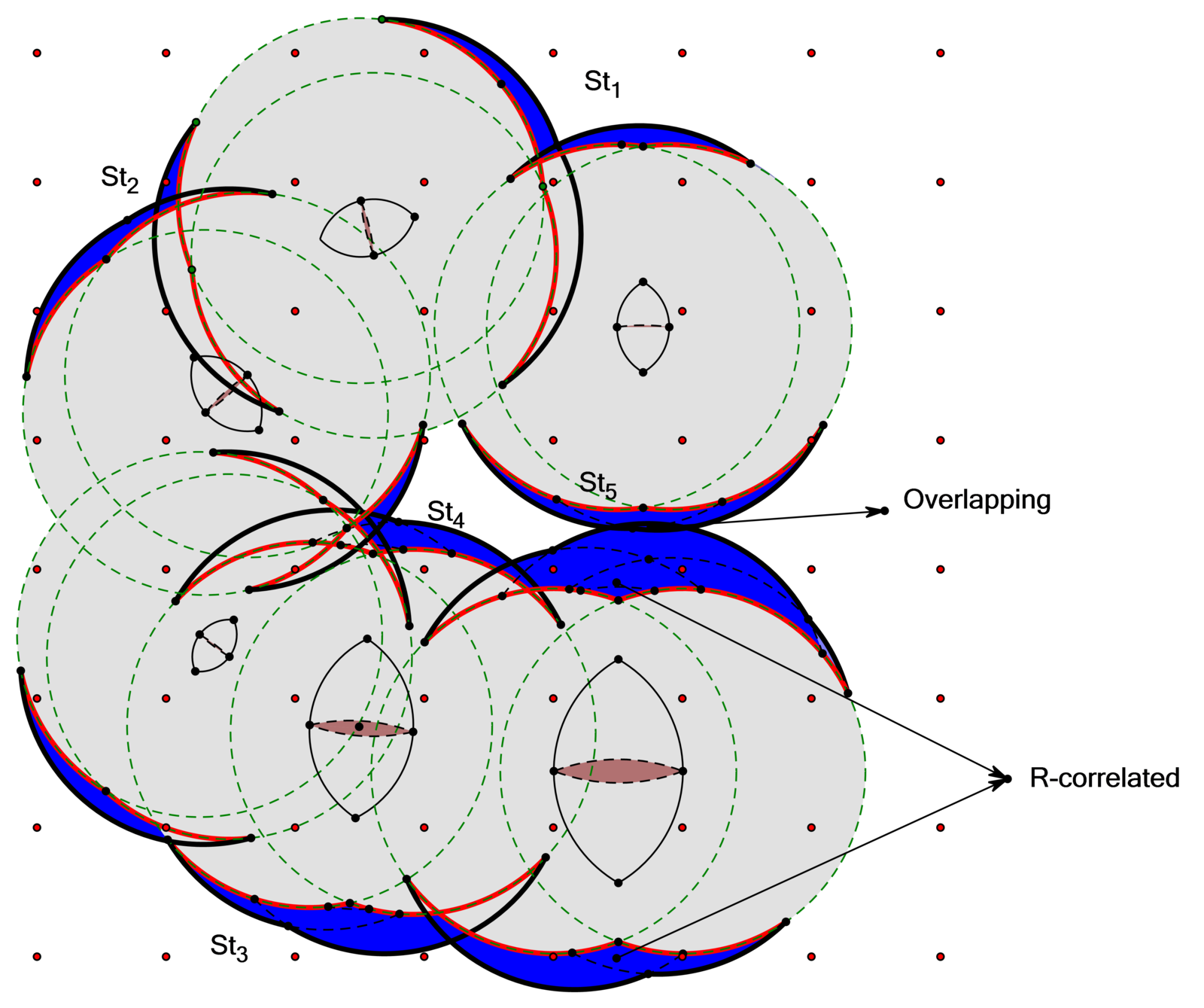}
  \caption{The general picture of the substructures in a block.
  The red points are the grid points of small squares. Dash green disks are
  what we have selected in
    $\calH$. There are five substructures in the block.
    }
  \label{fig:04:03}
\end{figure}

\begin{enumerate}

\item(Arcs)
Consider a disk $\Disk\in \setDisk'$ and suppose the center of $\Disk$ is in the square $\ssquare$.
Let $\Disk_s\Disk_t$ be the square gadget $\gadget(\ssquare)$, and
without loss of generality assume the line $\dcenter_s\dcenter_t$ is horizontal and
divides the plane into two halfplanes $H^+$ and $H^-$.
$\Disk$ may contribute at most two
\emph{uncovered arcs}, one in $H^+$ and one in $H^-$.
Let us first focus on $H^+$. $\partial \Disk$ intersects $\partial \calH$
at several points \footnote{
The number must be even.
}
in $H^+$. The uncovered arc is the segment of $\partial \Disk$ starting from the first intersection point
and ending at the last intersection point.
\footnote{
Note that an uncovered arc may not entirely lie in
the uncovered region $\Uncover$
(some portion may be covered by some disks in $\calH$).
}
We can define the uncovered arc for $H^-$ in the same way (if $|\partial \Disk \cap \partial \calH|\ne 0$).
Figure~\ref{fig:arc} illustrates why we need so many words to define an
arc. Essentially, some portions of an arc may be covered by some other
disks in $\calH$, and the arc is broken into several pieces.
Our definition says those pieces should be treated as a whole.
In this paper, when we mention an \emph{arc}, we mean an entire uncovered arc
(w.r.t. the current $\calH$).
Note that both endpoints of an arc lie on the boundary of $\calH$.

\item(Subarcs)
For an arc $a$, we use $a[A,B]$ to
denote the closed subarc of arc $a$ from point $A$ to point $B$.
Similarly, we only write $a(A,B)$ to denote the corresponding open subarc
(with endpoints $A$ and $B$ excluded).

\item(Central Angle)
Suppose arc $a$ is part of $\partial \Disk $ for some disk $\Disk$ with center $\dcenter$.
The central angle of $a$, denoted as $\angle(a)$ is the angle whose apex (vertex) is $\dcenter$
and both legs (sides) are the radii connecting $\dcenter$ and the endpoints of $a$.
We can show that $\angle(a)<\pi$ for any arc $a$ (See Lemma~\ref{lm:pi} in Appendix~\ref{apd:substructure})

\item(Baseline) We use $\partial \calH$ to denote to be the boundary of $\calH$.
Consider an arc
$a$ whose endpoints $P_1, P_2$ are on $\partial \calH$. We say the arc $a$ cover a point $P\in \partial \calH$,
if $P$ lies in the segment between $P_1$ and $P_2$ along $\partial \calH$.
We say a point $P\in \partial \calH$ {\em can be covered} if some arc in $\setDisk'$ covers $P$.
A baseline is a consecutive maximal segment of $\partial \calH$ that can be covered.
We usually use $\baseline$ to denote a baseline.

\item(Substructure)  A substructure $\substructure(\baseline,\calA)$ consists of
a baseline $\baseline$ and the collection $\calA$ of arcs which can cover some point in $\baseline$.
The two endpoints of each arc $a\in \calA$ are on
$\baseline$ and $\angle(a)$ is less than $\pi$.
Note that every point of $\baseline$ is covered by some arc in $\calA$.
Figure~\ref{fig:substructure} illustrates the components of an substructure.

Occasionally, we need a slightly generalized notion of substructure.
For a set $\calA$ of uncovered arcs,
if they cover a consecutive segment of the boundary of $\calH$,
$\calA$ also induces a substructure denoted as $\substructure[\calA]$.


\end{enumerate}

\begin{figure}[t]
  \centering
  \includegraphics[width=0.3\textwidth]{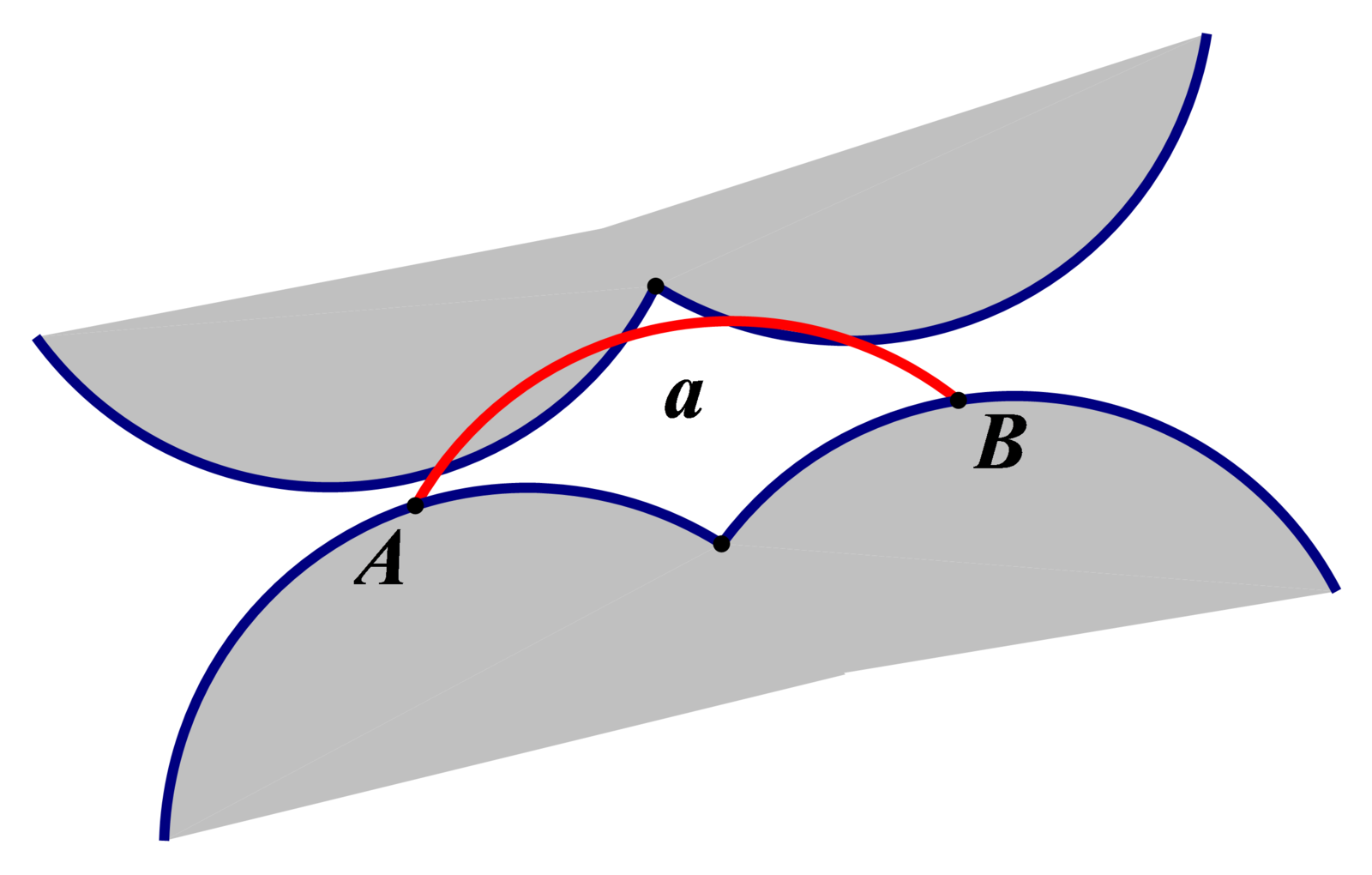}
  \caption{The figure gives an example of an arc. The blue curves are part of the boundary of
    $\calH$. The red curve is an uncovered arc.}
  \label{fig:arc}
\end{figure}

\begin{figure}[t]
  \centering
  \includegraphics[width=0.45\textwidth]{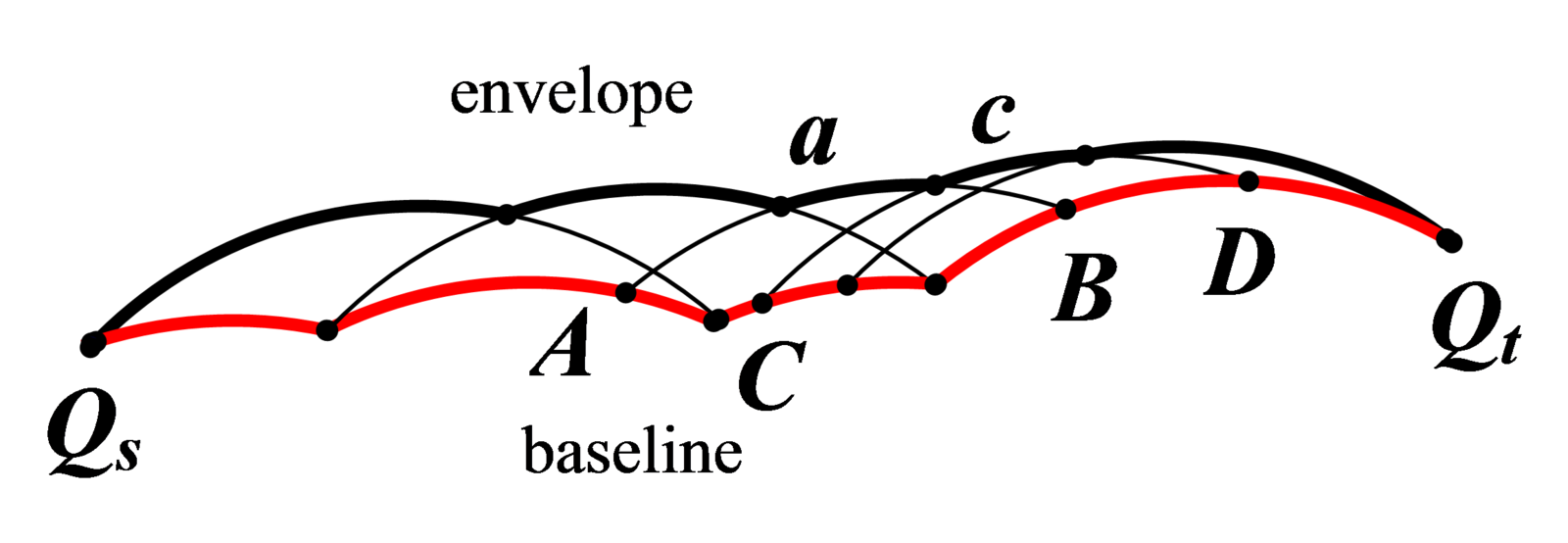}
  \caption{A substructure. The  baseline $\baseline$ consists of the red arcs which are the part of
    consecutive boundary of  $\partial \calH$. $Q_s, Q_t$ are the endpoints of
    $\baseline$. The black curves are are uncovered arcs.  The bold black arcs form the
    envelope. The arc $a \prec c$ because $A \prec C$ and  $B \prec D$.}
  \label{fig:substructure}
\end{figure}

\topic{Arc Order}
Now we switch our attention to the order of the arcs in a substructure $\substructure(\baseline,
\calA)$. Suppose the baseline $\baseline$ starts at point $Q_s$ and ends up at point $Q_t$. Consider
any two points $P_1$ and $P_2$ on the baseline $\baseline$. If $P_1$ is closer to $Q_s$ than
$P_2$ along the baseline $\baseline$, we say that $P_1$ \emph{appears earlier} than $P_2$
(denoted as $P_1 \prec P_2$).
Consider any two arcs $a$ and $c$ in $\calA$. The endpoints of arc $a$
are $A$ and $B$, and the endpoints of arc $c$ are $C$ and $D$. All of points $A, B, C, D$ are on the
baseline $\baseline$. Without loss of generality, we assume that $A \prec B$, $C \prec D$ and $A \prec
C$. If $B \prec D$, We say arc $a$ \emph{appears earlier} than  arc $c$
(denoted as $a \prec c$).
Otherwise, we say $a$ and $c$ are incomparable.
See Figure~\ref{fig:substructure} for an example.
It is easy to see that $\prec$ defines a partial order.

\topic{Adjacency}
Consider two arcs $a$ (with endpoints $A\prec B$) and $c$ (with endpoints $C\prec D$).
If $a \prec c$ and $C \prec B$,
we say that $a$ and $b$ are \emph{adjacent} (we can see that they must intersect exactly once),
and $c$ is the \emph{adjacent successor} of $a$.
Similarly, we can define the adjacent successor of subarc $a[P_1,P_2]$. If $c$ is the adjacent
successor of $a$, meanwhile $c$ intersects with
subarc $a[P_1, P_2]$, we say that $c$ is the \emph{adjacent successor} of subarc $a[P_1, P_2]$.
Among all adjacent successors of $a[P_1, P_2]$, we call the one whose
intersection with $a[P_1, P_2]$ is closest to $P_1$ the \emph{first adjacent successor} of
$a[P_1, P_2]$.

In order to carry out the dynamic program in Section~\ref{sec:DP},
we need to properly orient the baseline for each substructure so that the (partial) order or the arcs is well defined.
Our final solution in each substructure can be represented as a path
(which is a segment of the boundary of the union of chosen disks).
Our dynamic program essentially needs to determine such a path.
To be precise, we provide a formal
definition of a \emph{valid path}, as follows.

\begin{definition} [A Valid Path]
Consider a substructure $\substructure(\baseline, \calA)$.
Suppose the baseline $\baseline$ is
oriented from $Q_s$ to $Q_t$.  A valid path is a path from $Q_s$ to $Q_t$ which consists of a
sequence of subarcs $\{a_1[Q_s, Q_1]$, $a_2[Q_1, Q_2]$, $\ldots$, $a_k[Q_{k-1},Q_t] \}$
(baseline segments are considered as subarcs as well).
For any $a_i$, $a_{i+1}$ is its adjacent
successor (so $a_i\prec a_{i+1}$).
$Q_i$ is the intersection points of the arcs $a_i$ and $a_{i+1}$.
\end{definition}

Note that the baseline from $Q_s$ to $Q_t$ is a trivial valid path
(we do not consider any coverage requirement yet).
Among all the valid paths in a substructure, there is one that
is maximal in terms of the coverage ability, which we call {\em the envelope} of the substructure.

\begin{definition}[Envelope of a Substructure]
Consider a substructure $\substructure(\baseline, \calA)$.
Suppose the baseline $\baseline$ is
oriented from $Q_s$ to $Q_t$.
The envelope of $\substructure $ is the valid path
$\{a_1[Q_s, Q_1]$, $a_2[Q_1, Q_2]$, $\ldots$, $a_k[Q_{k-1},Q_t] \}$ where
$a_{i+1}$ is the first adjacent successor of $a_{i}$ for all $i \in [k]$.
\end{definition}


\topic{Coverage}
Consider a substructure
$\substructure(\baseline, \calA)$.
Consider an arc $a$ with endpoints $A$ and $B$ on baseline $\baseline$.
We use $\baseline[A,B]$ to denote the segment of $\baseline$ that is covered by $a$.
We say that the region surrounded by the arc $a$ and $\baseline[A,B]$ is \emph{covered} by arc $a$
and use $\Reg(a)$ to denote the region.
We note that the covered region $\Reg(a)$
is with respected to the current $\calH$.
Similarly, consider a valid path $\Path$. The region covered by $\Path$
is $\cup_{a\in \Path} \Reg(a)$ (the union is over all arcs in $\Path$)
and is denoted by $\Reg(\Path)$.
Finally, we define the region covered by the substructure $\substructure$,
denoted by $\Reg(\substructure)$, to be the region covered by the envelope of $\substructure$.

\section{Simplifying the Problem}
\label{sec:04:02}

The substructures may overlap in a variety of ways. As we mentioned in Section~\ref{sec:overview},
we need to include in $\calH$ more disks in order to make the substructures amenable to the dynamic
programming technique.  However, this step is somewhat involved and we decide to postpone it to the
end of the paper (Section~\ref{sec:calH}).  Instead, we present in this section what the
organization of the substructures  and what properties we
need  {\em after} including more disks in $\calH$ for the final dynamic program.

\topic{Self-Intersections}
In a substructure $\substructure$, suppose there are two arcs $a$ and $c$ in $\calA$ with endpoints
$A, B$ and $C, D$ respectively.
If $A \prec B \prec C \prec D$ and $a$ and $ c $ cover at least
one and the same point in $\calP$, we say
the substructure is \emph{self-intersecting}.
In other words, there exists at least one point covered by two non-adjacent arcs in a
self-intersecting substructure.
See Figure~\ref{fig:point-cut-off} for an example.
Self-intersections are troublesome obstacles for the dynamic programming approach.
So we will eliminate all
self-intersections in Section~\ref{sec:calH}.
In the rest of the section, we assume all
substructures are \emph{non-self-intersecting} and discuss their properties.

\begin{lemma}[Single Intersection Property]
  \label{lemma:singleIntersection}
  For any two arcs in a non-self-intersecting substructure, they have at most one intersection in the interior of the
  substructure.
\end{lemma}

\begin{proof}
  We prove by contradiction. Suppose $a_i$ and $a_j$ belong to the same substructure. $a_i$ and
  $a_j$ intersect at point $A$ and $B$. Since the substructure is non-self-intersecting, $a_i$ and
  $a_j$ lie on the same side of line $A$ and $B$. Because the two radii of $a_i$ and $a_j$ are
  equivalent, the sum of central angles of $a_i$ and $a_j$ equals $2\pi$. Thus at least one central
  angle of $a_i$  and $a_j$ is no less than $\pi$, rendering a contradiction to the fact that the
  central angle (defined in Sec.~\ref{sec:substructure}) of any arc is less than $\pi$.
 \end{proof}

Base on the single intersection property, we can easily get the following property.

\begin{lemma}
  In a non-self-intersecting substructure, if a point is covered by two arcs $a$ and $b$,
  $a$ is adjacent to $b$.
\end{lemma}

\topic{Order Consistency}
There are two types of relations between substructures which affect how the
orientations should be done.
One is the \emph{overlapping relation} and the other is \emph{remote-correlation.}
See Figure~\ref{fig:04:03}
for some examples.

As we alluded in Section~\ref{sec:overview}, the two substructures  which contain different related active
regions of the same gadget interact with each other, i.e.,
\begin{definition}[Remote correlation]
  Consider two substructures $\substructure_u$ and $\substructure_l$ which  are not
  overlapping. They contain two different active regions of the same gadget respectively
  (recalling that one gadget may have two different active regions, one in $H^+$, one in $H^-$). We say that they are
  \emph{remotely correlated} or \emph{R-correlated}.
  See Figure~\ref{fig:04:03}.
\end{definition}

There are two possible
baseline orientations for each substructure (clockwise or anticlockwise around the center of the
arc), which gives rise to four possible ways to orient both $\substructure_u$ and $\substructure_l$.
However, there are only two (out of four) of them are consistent (thus we can do dynamic programming
on them).  More formally, we need the following definition:

\begin{definition}[Local Order Consistency]
Consider two substructures $\substructure_u(\baseline_u, \calA_u)$ and $\substructure_l(\baseline_l,
\calA_l)$. There is an orientation for each substructure, such that the
partial orders of the disks for both substructures are consistent in the following sense: It can not
happen that $a_i \prec b_i$ in substructure $\substructure_u$ but $a_j \prec b_j$ in
$\substructure_l$, where $a_i, a_j \in \partial \Disk_a, b_i, b_j \in \partial \Disk_b$ and $ a_i,
b_i \in \calA_u, a_j, b_j \in \calA_l$.
\end{definition}

We show in the following simple lemma that the local order consistency
can be easily achieved for the two substructures containing different active regions of the same
gadget.

\begin{lemma}[Local Order Consistency]
  \label{lemma:04:03}
Consider two substructures $\substructure_u$ and $\substructure_l$ which are R-correlated. Each of
them contains only one active region. Then the local order consistency holds for the two
substructures $\substructure_u$ and $ \substructure_l$.
\end{lemma}

\begin{proof}
 We consider two substructures $\substructure_u$ and $\substructure_l$.
 The arcs not in the active regions have no influence on the order consistency
 since each of them only appears in one substructure.
 So, we only need to consider the order of the arcs in two active regions.
 We use the same notations as those on the RHS of Figure~\ref{fig:gadget}.
 We orient the upper baseline from $Q_s$ to  $Q_t$,
 and the lower baseline from $P_s$ to $P_t$.
Suppose two arcs $c_1, c_3$ belong to disk $\Disk_u$, and two arcs $c_2, c_4$
belong to the disk $\Disk_v$.
Assume $c_4 \prec c_3$ in substructure $\substructure_u$ and $c_1
\prec c_2 $ in substructure $\substructure_l$.
There must exist another intersection point on
each side of the line connecting the two intersections of $(c_1,c_2), (c_3,c_4)$.
This contradicts the fact that two unit disks have at most two intersections.
\end{proof}

Then, we discuss the situation where  two substructures overlap. We first need a few notations. For
an arc $a$, we use $\Disk(a)$ to denote the disk associated with $a$.  For a substructure
$\substructure(\baseline, \calA)$, we let $\Disk(\baseline)$ be the set of disks that contributes an
arc to the baseline $\baseline$.

\begin{definition}[Overlapping Relation]
\label{def:04:02}
  Consider two substructures $\substructure_1(\baseline_1,\calA_1)$ and
$\substructure_2(\baseline_2, \calA_2)$ and the point set $\calP$. We say $\substructure_1 $ and
$\substructure_2$ overlap when there are two arcs $a \in \calA_1$ and $ b \in \calA_2$ such
that both $a$ and $b$ can cover at least one and the same point of $\calP$ .
\end{definition}

%
%

%
%

The dynamic program need the following property. 
\begin{proposition}
The orientations of the two overlapping substructures should be  different (i.e.,  if one is clockwise, the
other should be  anticlockwise).
\end{proposition}

As different substructures may interact with each other,
we need a dynamic program which can run over all substructures simultaneously.
Hence, we need to define a globally consistent ordering of all arcs.

\begin{definition}[Global Order Consistency]
We have global order consistency if there is a way to orient the baseline of each substructure, such
that the partial orders of the disks for all substructures are consistent in the following sense: It
can not happen that $a_i \prec b_i$ in substructure $\substructure_i(\baseline_i, \calA_i)$ but $a_j
\prec b_j$ in $\substructure_j(\baseline_j, \calA_j)$, where $a_i, a_j \in \partial \Disk_a, b_i,
b_j \in \partial \Disk_b$ and $ a_i, b_i \in \calA_i, a_j, b_j \in \calA_j$.
\end{definition}

\topic{Substructure Relation Graph  $\auxgraph$}
we construct an auxiliary graph $\auxgraph$, called the \emph{substructure relation graph},
to capture all R-correlations and overlapping relations.
Each node in $\auxgraph$
represents a substructure.
If two substructures are
R-correlated, we add a blue edge between the
two substructures.  If two substructures overlap,
we add a red edge.

Consider a red edge between $\substructure_1(\baseline_1, \calA_1)$ and
$\substructure_2(\baseline_2, \calA_2)$. If baseline $\baseline_1$ is oriented clockwise (around
the center of any of its arc), then $\baseline_2$ should be oriented counterclockwise, and vise
versa.  The blue edge represents the same orientation relation, i.e., if $\substructure_1(\baseline_1, \calA_1)$ and
$\substructure_2(\baseline_2, \calA_2)$ are R-correlated, $\baseline_1$ and $\baseline_2$ should
be oriented differently.

It is unclear how to orient all baselines if $\auxgraph$ is an arbitrary graph.
So we need to ensure that $\auxgraph$ has a nice structure.
\begin{definition}  [Acyclic 2-Matching]
  \label{auxgraph}
  We say the substructure relation graph $\auxgraph$ is
  an acyclic 2-matching, if $\auxgraph$ is acyclic and is composed by
  a blue matching and a red matching.
  In other words, $\auxgraph$ contains only paths, and the red edges and blue
  edges appear alternately in each path.
\end{definition}

If $\auxgraph$ is a acyclic 2-matching, we can easily
assign each substructure a global arc order consistent orientations.

\topic{Point-Order Consistency}
Similarly to the arc order consistency, we also need define the \emph{point-order
  consistency}, which is also crucial for our dynamic program.
\begin{definition}[Point Order Consistency]
  Suppose a set $\calP_{\mathrm{co}}$ of points is covered by both of two overlapping substructures
$\substructure_1(\baseline_1, \calA_1)$ and $\substructure_2(\baseline_2, \calA_2)$. Consider any
two points $ P_1, P_2 \in \calP_{\mathrm{co}}$ and four arcs $a_1, a_2 \in \calA_1, b_1, b_2 \in \calA_2$. Suppose
$P_1 \in \Region(a_1) \cap \Region(b_1) $ and $P_2 \in \Region(a_2) \cap \Region(b_2)$. But $P_1
\notin \Region(a_2) \cup \Region(b_2)$ and $P_2 \notin \Region(a_1) \cup \Region(b_1)$.
We say $P_1$ and $P_2$ are point-order consistent if
$ a_1 \prec a_2 $ in $\substructure_1$ and $ b_1 \prec b_2 $ in $\substructure_2$.
We say the points in $\calP_{\mathrm{co}}$ satisfy point order consistency
if all pairs of points in $\calP_{\mathrm{co}}$ are point-order consistent.
\end{definition}

After introducing all relevant concepts,
we can finally state the set of properties we need for the dynamic program.

\begin{lemma}
\label{lm:calH}
After choosing $\calH$, we can ensure the following properties holds:
\begin{itemize}
\item [P1.]
(Active Region Uniqueness) Each substructure contains at most one active region.
\item [P2.]
(Non-self-intersection) Every substructure is non-self-intersecting.
\item [P3.]
(Acyclic 2-Matching)  The substructure relation graph $\auxgraph$ is an
acyclic 2-matching, i.e., $\auxgraph$ consists of only paths.
In each path, red edges and blue edges appear alternately.
\item [P4.]
(Point Order Consistency) Any point is covered by at most two substructures. The points satisfy
the point order consistency.
\end{itemize}
\end{lemma}

How to ensure all these properties will be discussed in detail in Section~\ref{sec:calH}.
Now, everything is in place to describe the dynamic program.

\section{Dynamic Programming}
\label{sec:DP}

Suppose we have already constructed the set $\calH$ such that Lemma~\ref{lm:calH} holds
(along with an orientation for each substructure).
Without loss of  generality, we can assume that the remaining disks can
cover all remaining points (otherwise, either the original instance is infeasible or our guess is wrong).
In fact, our dynamic program is inspired, and somewhat similar to
those in \cite{Ambuhl06, Erlebach10, Peled12}.

\topic{DP for Two Overlapping Substructures}
For ease of description, we first handle the case where there are only two overlapping
substructures. We will extend the DP to the general case shortly. 
Suppose the two substructures are $\substructure_1(\baseline_1,\calA_1)$ and
$\substructure_2(\baseline_2,\calA_2)$, $\baseline_1$ is oriented from $P_s$ to $P_t$ and
$\baseline_2$ is oriented from $Q_s$ to $Q_t$ 

A state of the dynamic program is a pair $\state=(P, Q)$ where $P$ is an intersection point of two arcs
in substructure $\substructure_1$ and $Q$ is an intersection point of two arcs in substructure
$\substructure_2$.
Fix the state $\state=(P, Q)$ and consider $\substructure_1$.
Let $\basearcP$ and $\toparcP$ be the two arcs intersecting at $P$.
Suppose  $ \basearcP \prec \toparcP$ with endpoints $(A,B), (C,D)$ respectively. 
We call arc $\basearcP$ the 
\emph{base-arc} and  $\toparcP$ the \emph{top-arc} for point $P$.
\footnote{
If $P$ is the tail endpoint of an arc (so $P$ is on the baseline), $P$ only
has a base-arc (no top-arc), which is the baseline arc it lies on.
}
Our DP maintains that the base-arc is always paid in the subproblem, and
  intermediate state.

  \begin{figure}[t]
    \centering
    \includegraphics[width=0.8\textwidth]{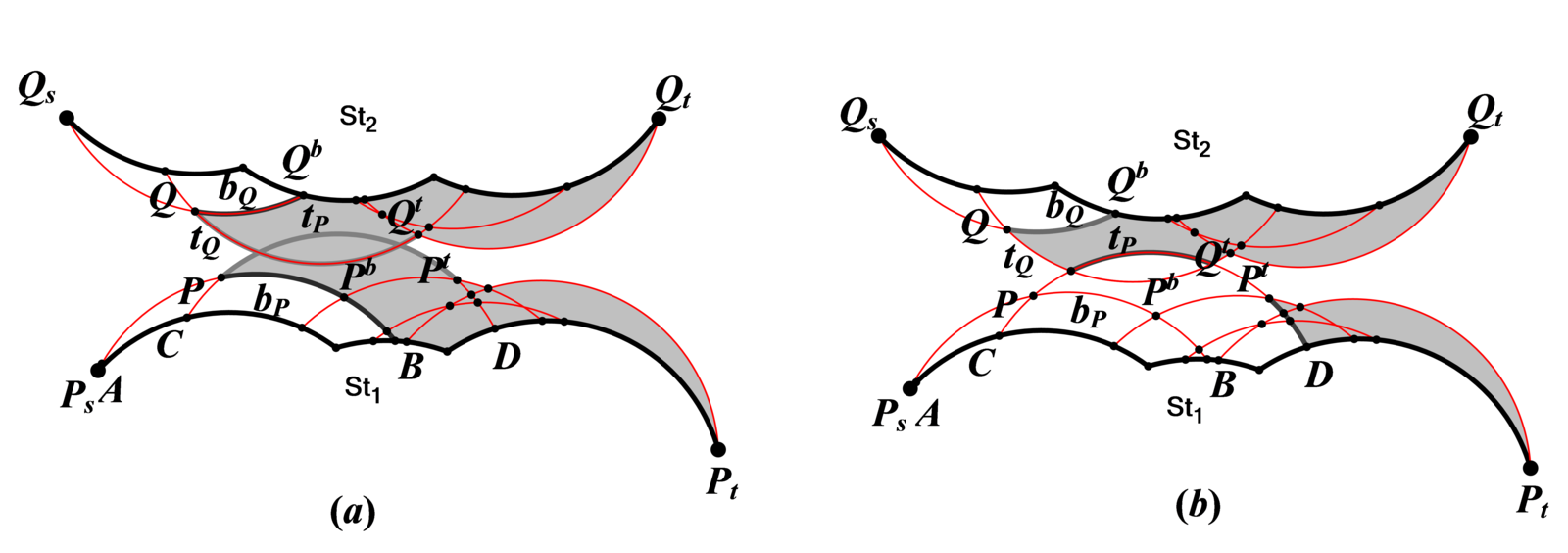}
    \caption{The figure explains the dynamic program of two overlapping substructures. The
      left figure shows the subproblem $\OPT(P,Q)$.  The goal of $\OPT(P,Q)$ is to find minimum valid
      paths for $PP_t$ and $QQ_t$ respectively in set $\calA_{1}[P]\cup \calA_2[Q]$ such that the paths cover
      all points of $\calP[P,Q]$.  The right figure illustrates one of its four smaller subproblems $\OPT(P^t, Q)$. }
    \label{fig:ffsubstructure}
  \end{figure}

Given state $\state=(P, Q)$, now we describe the subproblem associated with the state.
Again first focus on $\substructure_1$.
Intuitively, a feasible solution to the subproblem restricted to $\substructure_1$ (resp. $\substructure_2$)
is a valid path starting from point $P$ (resp. $Q$) and terminating at $P_t$ (resp. $Q_t$).
More specifically,
we construct a substructure $\substructure_1^{[P]}(\baseline_1[P], \calA_1[P])$ (See
Figure~\ref{fig:ffsubstructure}): 
\begin{itemize}
\item
$\baseline_1[P]$ is
the concatenation of subarc $\basearcP[P, B]$ and the original baseline segment $\baseline_1[B,P_t]$.
All arcs in $\baseline_1[P]$ have cost zero.
\item
$\calA_1[P]$ consists of all arcs $a'\in \calA_1$ such that $ b_P \prec a'$ (of course, with the
portion covered by $\baseline_1[P]$ subtracted).
The cost each such arc is the same as its original cost.
\end{itemize}
Similarly, we consider $\substructure_2$ and the intersection point $Q$,
and construct $\substructure_2[Q]$ with baseline $\baseline_2[Q]$ and arc set $\calA_2[Q]$.
We use
$\calP(a)$ (or $\calP(\calA)$) to denote the points covered by $a$ (or $\calA$)
(w.r.t. the original baseline).
Let the point set $\calP[P,Q]$ that we need to cover in the subproblem $\state(P,Q)$
be
\begin{align}
\label{eq:calP}
\calP[P,Q]=\calP(\calA_1[P])+\calP(\calA_2[Q]) - \calP(\basearcP) -\calP(\basearcQ).
\end{align}
We note that the minus term $ - \calP(\basearcP) -\calP(\basearcQ)$
is not vacuous as $\basearcP$ (resp. $\basearcQ$) may cover some points
in $\calA_2[Q]$ (resp. $\calA_1[P]$),
and it is important that we do not have to cover those point (this subtlety is crucial in the correctness proof of the DP).
The goal for the subproblem $\state(P,Q)$ is to find two valid paths with minimum total weight,
one from $P$ to $P_t$ and one from $Q$ to $Q_t$, such that
they together cover all points in $\calP[P,Q]$.
Note that the weights of both base-arcs $\basearcP$ and $\basearcQ$ should be included in any feasible solution as well.

Suppose $\basearcP(P,B]$ intersects its first successor at $P^b$ (called {\em base-adjacent point})
and $\toparcP(P,D]$ intersects its first successor at $P^t$ (called {\em top-adjacent point}).
Similarly, we define $Q^b$, $Q^t$ in
$\substructure_2$ in exactly the same way.

Now, computing the optimum for subproblem $\state(P,Q)$
reduces to computing the optima for four smaller subproblems
$ \OPT(P^b, Q) $,  $ \OPT(P^t,Q) $,
$ \OPT(P, Q^b) $ and $ \OPT(P, Q^t) $.
We define two Boolean variables $I_P$ (reps. $I_Q$) indicating
whether
we can move from $\calP_{P^b,Q}$ to $\calP_{P,Q}$ without choosing a new arc.
Formally, if $ \calP_{P,Q} = \calP_{P^b,Q}$, $I_P = 0$. Otherwise, $I_P = 1$.
Similarly,
if $\calP_{P,Q} = \calP_{P, Q^b} $, $I_Q = 0$. If not, $I_Q =1$.
The dynamic programming recursion is:
\begin{eqnarray}
\label{eq:dp}
  &&  \OPT(P,Q) =  \min  \left\{
\begin{array}{ll}
\OPT(P, Q^b)+ I_Q \cdot \infty, &\text{ add no new arc }  ;\\
\OPT(P, Q^t) + w[\basearcQ], &  \text{ add base-arc } \basearcQ  ; \\
\OPT(P^b,Q)+ I_P \cdot \infty, &\text{ add no new arc }  ;\\
\OPT(P^t,Q) + w[\basearcP], &  \text{ add base-arc } \basearcP . \\
\end{array}
\right.
\end{eqnarray}
The optimal value we return is $\OPT(P_s, Q_s)$.
Now, we prove the correctness of the DP in the following theorem.
We note that both the point-order consistency and the subtlety mentioned above
play important roles in the proof.

  \begin{theorem}
    \label{thm:dp}
Suppose that we have two overlapping substructures $\substructure_1(\baseline_1, \calA_1)$
and $\substructure_2(\baseline_2, \calA_2)$.
Further suppose that $\baseline_1$ and $\baseline_2$ are oriented in a way such that
the point-order consistency holds.
Then, the cost of the optimal solution equals to $\OPT(P_s, Q_s)$ (which is computed by \eqref{eq:dp}).
  \end{theorem}

  \begin{proof}
Consider subproblem $\OPT(P,Q)$.  As we know the optimal solution of $\OPT(P,Q)$ should be two
valid paths. One is from $P$ to $P_t$ and the other is from $Q$ to $Q_t$.
Suppose they are
$\Path_1=$ $\{a_1[P, P_1] $, $a_2[P_1, P_2] $, $\ldots$, $a_i[P_{i-1}, P_i]$, $\ldots $, $a_k[P_{k-1},P_t] \} $
and $\Path_2
=$ $\{ b_1[Q, Q_1]$, $b_2[Q_1, Q_2]$, $\ldots$, $b_i[Q_{i-1}, Q_{i}]$, $\ldots$, $b_l[Q_{l-1},Q_t] \} $.
We can see that it suffices to
prove that at least one of the two statements is true.
    \begin{itemize}
    \item The pair of paths $(\Path_1 - \{ a_1[P, P_1) \},\Path_2)$ is the optimal solution to
      subproblem $\OPT(P_1, Q)$.
    \item The pair of paths $(\Path_1, \Path_2 - \{ b_1[Q, Q_1) \})$ is the optimal solution to
      subproblem $\OPT(P, Q_1)$.
    \end{itemize}
    We prove by contradiction.
    Assume that both of the above statements are wrong.
    Suppose $\basearcP$ and $\basearcQ$ are the
    base-arcs for state $\state(P,Q)$, i.e., $ \basearcP $ intersects with $a_1$ at point $P$ and $\basearcQ$
    intersects with $b_1$ at point $Q$.
    We use $\calP(\Path)$ to denote the point set which is covered by
    path $\Path$.
    Recall that
    $\calP[P,Q] = \calP(\calA_1[P])+\calP(\calA_2[Q]) - \calP(\basearcP) - \calP(\basearcQ)$.
    Since $\Path_1\cup \Path_2$ is the optimal solution for $\OPT(P,Q)$ (hence feasible),
    we have that
    $\calP[P,Q] = \calP(\Path_1) + \calP(\Path_2) - \calP(\basearcP) - \calP(\basearcQ)$.
    Then, the pair of paths $(\Path_1 - \{ a_1[P, P_1) \},\Path_2)$
    is the optimal solution
    for the subproblem in which we need to pick one path from $P_1$ to $P_t$
    and one from $Q$ to $Q_t$ to
    cover the points in 
    $$
    \calP(\Path_1) + \calP(\Path_2) - \calP(\basearcP) - \calP(\basearcQ) -\calP(a_1)=\calP[P,Q]-\calP(a_1).
    $$
    If not, we can get a contradiction by replacing $\Path_1 - a_1[P, P_1] $ and $\Path_2 $
    with the optimal solution of the above subproblem,
    resulting in a solution with less weight than $\OPT(P,Q)$.
    Since the first statement is wrong,
    we must have that
    $$
    \calP[P,Q]-\calP(a_1) \ne \calP[P_1, Q]
    $$
    (otherwise $(\Path_1 - \{ a_1[P, P_1) \},\Path_2)$ is optimal for $\state(P_1,Q)$).
    We note that the LHS $\subseteq$ RHS.
    Plugging the definition~\eqref{eq:calP}, we have that
    $$
    \calP(\calA_1[P])+\calP(\calA_2[Q]) - \calP(\basearcP) -\calP(\basearcQ) -\calP(a_1)
    \ne
    \calP(\calA_1[P_1])+\calP(\calA_2[Q]) - \calP(a_1) -\calP(\basearcQ).
    $$
    A careful (elementwise) examination of the above inequality shows that it is only possible if
     \begin{eqnarray*}
     \calP(\basearcP)  \cap  \left(\calP(\calA_2[Q])  - \calP(\basearcQ)\right) \neq \emptyset.
  \end{eqnarray*}
  Repeating same argument, we can see that if the second statement is wrong, we have that
     \begin{eqnarray*}
     \calP(\basearcQ)  \cap  \left(\calP(\calA_1[P])  - \calP(\basearcP)\right) \neq \emptyset.
  \end{eqnarray*}
    Hence, there exist $b_i \in \calA_2[Q]$ and $a_j \in \calA_1[P]$ such that
  $
    \calP(\basearcP) \cap \calP(b_i) \neq \emptyset, \text{ and } \calP(\basearcQ) \cap \calP(a_j) \neq \emptyset.
  $
  However, this contradicts the point-order consistency because of $ \basearcP \prec a_j$ and $ \basearcQ
  \prec  b_i$.

  Thus, one of the two statements is true. W.l.o.g, suppose ($\Path_1 - \{a_1[P, P_1] \}, \Path_2$) is
  the optimal solution to subproblem $\OPT(P_1,Q)$. If $P_1$ is the top-adjacent point of $P$,
  through $\OPT(P, Q) = \OPT(P_1,Q) + w[\basearcP]$ in DP\eqref{eq:dp}, we can get the optimal
  solution of subproblem $\OPT(P,Q)$. If not, suppose $P^t, Q^t$ is the top-adjacent point of $P, Q$, using
  the same argument, we can prove that ($\Path_1, \Path_2$) is the optimal solution to at least one
  of the subproblems $\OPT(P, Q^t)$ or $\OPT(P^t, Q)$. Thus, we can get the optimal solution for
  $\OPT(P,Q)$ by our DP.

  \end{proof}

\topic{DP for the general problem}
Then, we  handle all substructures together.
Our goal is find a valid path for each substructure
such that minimizing the total weight of all the paths.
We can see that we only need to handle each path in the substructure relation graph $\auxgraph$ separately
(since different paths have no interaction at all).
Hence, from now on, we simply assume that $\auxgraph$ is a path.

Suppose the substructures are $\kset{\substructure_k(\baseline_k, \calA_k)}$.
We use $A_k$ and $B_k$ to denote two endpoints of $\baseline_k$.
Generalizing the previous section, a state for the general DP
is $\state=\kset{P_k}$, where
$P_k$ is an intersection point in
substructure $\substructure_k$.
We use $b_{P_k}$, $t_{P_k}$, $P^b_k$, $P^t_k$ to
denote the base-arc, top-arc, base-adjacent point, top-adjacent point (w.r.t. $P_k$) respectively.
For each $k\in [m]$, we also define $\substructure_k^{[P_k]}(\baseline_k[P_k],\calA_k[P_k])$
in exactly the same way as in the previous section.
Let
$\calP\left[\kset{P_k}\right]$ be the point set we need to cover in the subproblem:
$$
\calP\left[\kset{P_k}\right]=\bigcup_{k\in [m]}\calP(\calA_k[P_k]) - \bigcup_{k\in [m]} \calP(\basearcP_k).
$$
The subproblem $\OPT(\kset{ P_k })$ is to find, for each substructure $\substructure_k$,
a valid path from  $P_k$ to $B_k$,
such that all points in $\calP[\kset{P_k}]$ can be covered and the total cost is minimized.

The additional challenge for the general case is caused by R-correlations.
If two arcs (in two different substructures) belong to the same disk, we say that they
are {\em siblings} of each other.
If we processed each substructure independently, some disks would be counted twice.
In order to avoid double-counting, we should consider both siblings together,
i.e., select them together and pay the disk only once in the DP. 



In order to implement the above idea, we need a few more notations.
We construct an auxiliary bipartite graph $\diskarcgraph$.
The nodes on one side are all disks in $\setDisk'\setminus \calH$, and the nodes on the other side are substructures.
If disk $\Disk_i$ has an arc in the substructure $\substructure_j$, we
add an edge between $\Disk_i$ and $\substructure_j$. Besides, for each arc of baselines, we add a
node to represent it and add an edge between the node and the substructure which contains the
arc. Because the weight of any arc of baselines is zero, it  shall not induce contradiction that
regard them as independent arcs. 
In fact, there is a 1-1 mapping between the edges in $\diskarcgraph$ and all arcs.
See Figure~\ref{fig:bipartite} for an example.

Fix a state $\state=\kset{P_k}$.
For any arc $a$ in $\substructure_k$ (with intersection point $P_k$  and base-arc $\basearcP_k$),
$a$ has three possible positions:
\begin{enumerate}
\item $a \prec \basearcP_k$: we label its corresponding edge with ``unprocessed'';
\item $a= \basearcP_k$: we label its corresponding edge with ``processing'';
\item Others: we label its corresponding edge with ``done''.
\end{enumerate}
As mentioned before, we need to avoid the situation where one arc becomes the base-arc first
(i.e., being added in solution and paid once), and its sibling becomes the base-arc in a later step (hence being paid twice).
With the above labeling, we can see that all we need to do is to avoid the states in which one arc is ``processing'' and its
sibling is ``unprocessed''.
If disk $\Disk$ is incident on at least one ``processing'' edge and
not incident on any ``unprocessed'' edge, we say the $\Disk$ is {\em ready}.
Let $\ReadyDisk$ be the set of ready disks.
For each ready disk $\Disk$, we use $\Neighbor$ to denote the set of neighbors (i.e., substructures)
of $\Disk$ connected by ``processing'' edges.
We should consider all substructures in $\Neighbor$ together.



  \begin{figure}[t]
    \centering
    \includegraphics[width=0.65\textwidth]{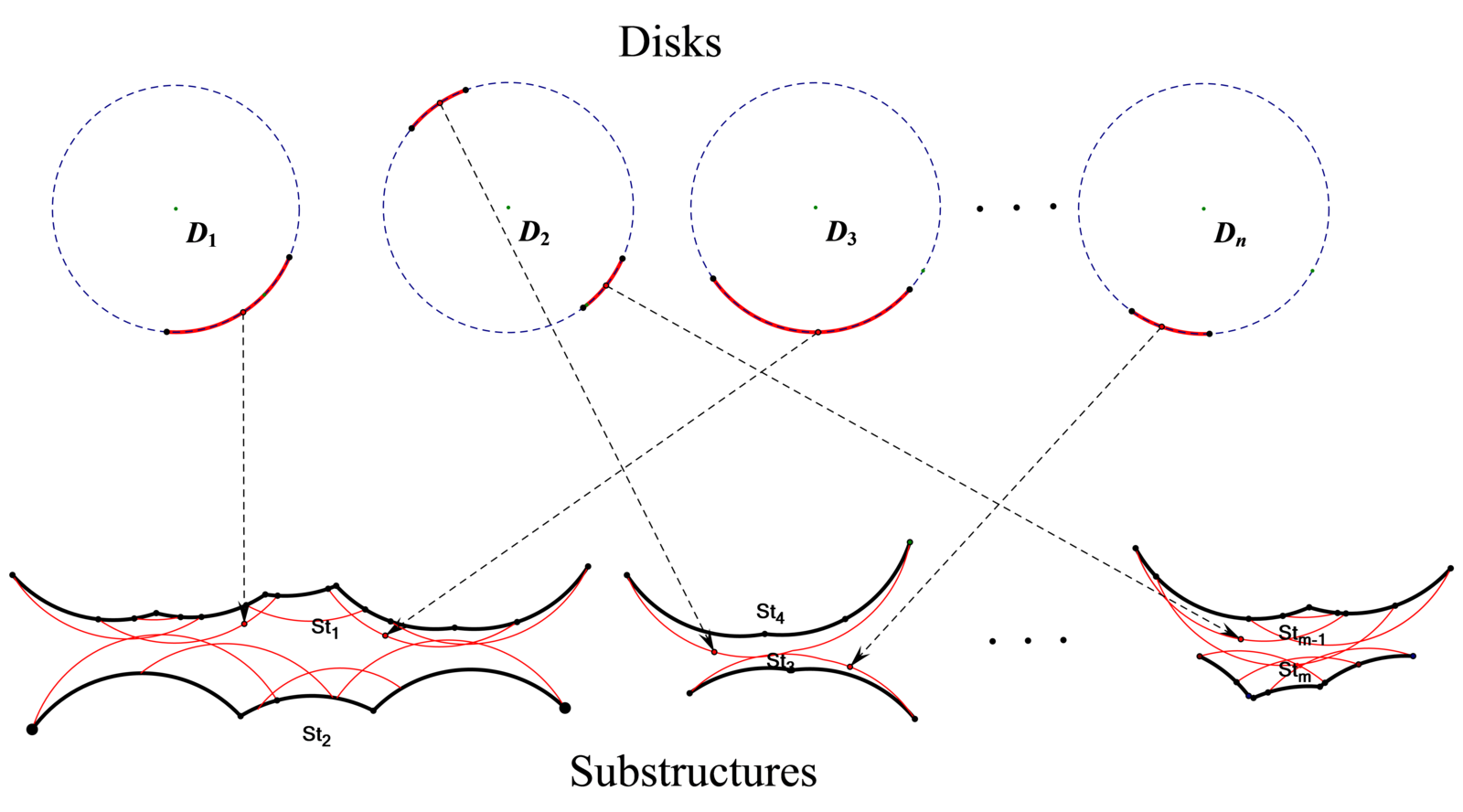}
    \caption{The bipartite graph which is used for marking the \emph{ready} disks. The nodes
      on upper side represent the disks. The nodes on the lower side represent the
      substructures. If $\Disk_i$ has an arc in $\substructure_j$, we add an arc between them. }
    \label{fig:bipartite}
  \end{figure}

Again, we need in our DP
indicator variables to tell us whether a certain transition is feasible:
Formally, if 
$\calP[\kset{P_k}] = \calP[[\calP_k][P^b_i]_{\{i\}}]$, let $ I_i = 0 $.
Otherwise, let $ I_i=1 $.
For ease of notation, for a set $\kset{e_k}$ and $S\subseteq [m]$,  we write
$[e_k][e'_i]_S=\{e_k\}_{k\in [m]\setminus S} \cup \{e'_i\}_{i\in S}.$
Hence, 
$$
[P_k] [P_i^b]_{\{i\}}= \{P_k \}_{k\in[m]\backslash i} \cup P_i^b
\text{ and }
$$
$$
[P_k][P_i^t]_{\Neighbor}=\{ P_k \}_{k\in[m] \setminus \Neighbor}  \cup  \{ P_i^t \}_{i \in \Neighbor}
$$
Then we have the dynamic program as follows:
\begin{eqnarray}
\label{dp2}
&& \OPT\left(\kset{P_k} \right) =
  \min\left\{  \begin{array}{ll}
\min_{i \in [m]} \left\{\OPT\left([P_k] [P_i^b]_{\{i\}} \right) + I_i \cdot \infty  \right\},
& \text{ add no disk}\\
\min_{\Disk \in \ReadyDisk } \left\{\OPT\left([P_k][P_i^t]_{\Neighbor}\right) +
w_{\Disk} \right\},
& \text{ add disk  } \Disk \\
\end{array}
\right. \\ \nonumber
\end{eqnarray}
Note that in the second line, the arc(s) in $\Neighbor$ are base-arcs (w.r.t. state $\state(\kset{P_k})$.

In the rest of the section, we prove the correctness of the dynamic program.
If we use the
solution of a smaller subproblem $\OPT(\state')$ to compute subproblem $\OPT(\state)$, we say $\state$ can be reached
from $\state'$ (denoted as $\state' \to \state$).
If $\state$ can be reached from initial
state $\state_0(\kset{P_k \mid  P_k = B_k})$, we say the state is reachable, which is denoted by
$\state_0 \to \state$.

We start with a simple consequence of our DP: there is no double-counting.

\begin{lemma}
  If  two arcs in the solution belong to the same disk, their weights are counted only once in
  the DP \eqref{dp2}.
\end{lemma}
\begin{proof}
  From the DP, we can see that the weight of an arc is counted only
  when it becomes a base-arc (its label changes from ``unprocessed'' to ``processing'').
  If the two sibling arcs $a, b$ (belonging to disk $\Disk$) in the solution are counted twice, there exist two
  states $\state_1$ and $\state_2$ such that (1)$\state_1 \to \state_2$, (2)$a $ is a base-arc in
  $\state_1$, but $b$ is not a base-arc in $\state_1$,
  (3) $b$ is a base-arc in $\state_2$.
  So, in $\state_1$ or any state before that, arc $b$ is ``unprocessed''.
  However, $a$ can become a base-arc only when $\Disk$ is ready, rendering a contradiction.
\end{proof}

\eat{

\begin{lemma}
  \label{lm:neighbor}
  Suppose $\state$ is a reachable state, and $\ReadyDisk$ is defined with respect to $\state$.
  Then, we have that $\bigcup_{\Disk \in \ReadyDisk} \Neighbor = [m]$.
\end{lemma}
\begin{proof}
  Obviously the initial state $(\kset{P_k \mid  P_k = B_k})$
  satisfies the lemma (the ``processing'' edge for $P_k$ is the baseline arc $P_k$ lies on).
  We prove by contradiction.  Suppose there exists an reachable state
  $\state$ which does not satisfy the property. It means there exist a $i$ such that $i \in [m] -
  \bigcup_{\Disk \in \ReadyDisk} \Neighbor, P_i \neq B_i$. According to our dynamic program,
  all  state which can reach the $\state$ should contains $P_i$.  Thus, $\state_0 \nrightarrow
  \state$. It contradicts to the fact that $\state$ is reachable.
\end{proof}
}

Now, we prove the correctness of the dynamic program.
The proof is a generalization of Theorem~\ref{thm:dp}.

\begin{theorem}
  \label{thm:dpgeneral}
  Suppose $\auxgraph$ is a path.
  All baselines are oriented such that all properties in Lemma~\ref{lm:calH} hold.
  Then, the optimal cost for the problem
  equals to $\OPT(\kset{A_k})$ (computed by our DP~\eqref{dp2}).
\end{theorem}

\begin{proof}
  Suppose the set $\kset{\Path_k}$ of paths is the optimal solution.
  We need to prove
  (1) the final state $\state(\kset{A_k})$ is
  reachable,
  (2) for any reachable state $\state=\kset{P_k}$,
  $\OPT(\kset{P_k})$ computes the optimal solution for the corresponding subproblem
  (that is to find one valid path from $P_k$ to $B_k$ for each substructure $\substructure_k^{[P_k]}$ to cover all point in
  $\calP[\kset{P_k}]$, such that the total cost is minimized).

  We first prove the first statement.
  Suppose the state $\state$ is reachable, we prove it can reach another
  state if $\state$ is not the final state (i.e., we do not get stuck at $\state$).
  Let us prove it by contradiction. Assume we get stuck at state $\state$.
  That means there is no ready disk in $\state$.
  Note that each substructure, say $\substructure$, is incident on exactly one ``processing'' arc, say arc $a$
  (which is the base-arc in $\substructure$).
  $a$'s sibling, say $b$ (in $\substructure'$), must be labeled ``unprocessed'' (otherwise the disk would be ready).
  Consider the base-arc (or ``processing'' arc), say $a'$, in $\substructure'$.
  So we have $a'\prec b$ in $\substructure'$.
  Again the sibling $b'$ of $a'$ must be an ``unprocessed'' arc in $\substructure$.
  \footnote{To see that $b'$ is in $\substructure$, note that $\auxgraph$ is a path
  and $\substructure$  is only R-correlated with $\substructure'$ (and vice versa).
  }
  So we have $b'\prec a$ in $\substructure$, which contradicts the global arc-order consistency.

  Now, we prove the second statement.
  We consider state $\state = \kset{P_k}$.
  Suppose the optimal solution for subproblem $\state\kset{P_k}$ is the set of paths
  $\Pathset =$ $\kset{\Path_{k}}$
  where $\Path_k=(a_{k_1}, a_{k_2} \ldots, a_{k_n})$.
  Consider
  the states $\state_{\Neighbor}:=[P_k][P_i^t]_{\Neighbor}$
  for all $\Disk \in  \ReadyDisk$.
  Obviously, we can see from our DP that $\state_{\Neighbor} \to \state$.
  Define for each $\Disk \in  \ReadyDisk$, a set of paths
  $$
  \Pathset_{\Neighbor} = \{\Path_k
  \}_{k\in[m]- \Neighbor} \cup \{ \Path_{i} - \{a_{i_1}\} \}_{i \in \Neighbor}.
  $$
  It suffices to prove that there exists at least one
  $\Disk \in \ReadyDisk$
  such that $\Pathset_{\Neighbor}$ is the optimal solution for $\OPT(\state_{\Neighbor})$.

  Consider a substructure $\substructure_i$.
  Suppose the intersection point in $\substructure_i$ of $\state$ is $P_i$ and the base-arc at
  point $P_i$ is $\basearcP_i$.
  For each $i\in [m]$, let $\substructure_{\match(i)}$ be 
  the only substructure (if any) overlapping with $\substructure_i$.
  So $\basearcP_{\match(i)}$ is the base-arc of $\substructure_{\match(i)}$.
  Using exactly the same exchange argument in
  Theorem~\ref{thm:dp}, we can show that 
  if $\Pathset_{\Neighbor}$ is not the optimal solution
  for $\OPT(\state_{\Neighbor})$,  
  there exists some $i \in \Neighbor$ such that $\event_i$ happens, 
  where $\event_i$ is the following event: there exists 
  an arc $\beta_{\match(i)}$ in $\substructure_{\match(i)}$ 
  with $\beta_{\match(i)} \succ \basearcP_{\match(i)}$
  such that 
  $$
  \calP(\basearcP_i) \cap \calP(\beta_{\match(i)}) \neq \emptyset.
  $$
We use $\event_i$ to denote the above event. 
Thus if there is no $\Disk\in \ReadyDisk$ such that $\Pathset_{\Neighbor} $
  is the optimal solution for $\OPT(\state_{\Neighbor})$, we have
\begin{equation}
  \label{eq:clause}
  \bigwedge_{\Neighbor \in \ReadyDisk} \left( \bigvee_{i\in \Neighbor} \event_i\right) = \true.
\end{equation}
Converting the conjunctive normal form (CNF) to the disjunctive normal form (DNF),
we get
\begin{equation*}
  \bigvee_{(k_1,\ldots,k_{|\ReadyDisk|}) \in \Pi_{\Disk \in \ReadyDisk}\Neighbor } \left(\bigwedge_{i \in |\ReadyDisk|} \event_{k_i} \right)=\true,
\end{equation*}
where $\Pi_{\Disk \in \ReadyDisk}\Neighbor$ means the Cartesian product of all $\Neighbor$ in $\ReadyDisk$. 
We call each
$\bigwedge_{i \in |\ReadyDisk|} \event_{k_i}$ a clause (note that $k_i$ indexes a substructure). 
If we can prove that every clause is false,
then obviously,  \eqref{eq:clause} is wrong, resulting in a contradiction.

Now, we show that every clause is false.  First, we consider the case that both end nodes of $\auxgraph$
are incident to red edges.  W.l.o.g., suppose the two nodes of $\auxgraph$ are $\substructure_1$ and
$\substructure_2$.  Thus, they are R-correlated with other substructures.  We know if one
substructure $\substructure_i$ is not R-correlated with others, every clause must contain the
corresponding event $\event_i$ (since the disk corresponding to its base-arc must be ready and in
$\ReadyDisk$).  Hence, every clause contains $\event_1$ and $\event_2$.  Moreover, for each pair
$(\substructure_i, \substructure_j)$ of R-correlated substructures, each clause should contain
either $\event_i$ or $\event_j$.  Suppose the length of the path $\auxgraph$ is $\ell$.  Because red
and blue edges alternates, there are $\frac{\ell-1}{2}$ R-correlated substructure pairs and
$\frac{\ell+1}{2}$ overlapping substructure pairs.  We should select $\frac{\ell-1}{2} + 2$ terms in
each clause.  Because $\frac{\ell-1}{2}+ 2 > \frac{\ell+1}{2}$, there exists a pair of overlapping
substructures $(\substructure_i, \substructure_{\match(i)})$ such that both $\event_i$ and
$\event_{\match(i)}$ appear in the clause.  To make the clause true, we must have
$$
\event_i= (\calP(\basearcP_i) \cap \beta_{\match(i)} \neq \emptyset)=\true \quad\text{ and }\quad
\event_{\match(i)}= (\calP(\basearcP_{\match(i)}) \cap \beta_{i} \neq \emptyset )=\true,
$$
where $\beta_{\match(i)} \succ \basearcP_{\match(i)}$ and $\beta_{i} \succ \basearcP_{i}$.  It
yields a contradiction to the point-order consistency.

Next, we consider the remaining case where at least one end of the path $\auxgraph$ is a blue edge,
meaning the substructure on the end does not overlap with any other substructure. W.l.o.g., suppose
the end node is $\substructure_1$ and it is R-correlated with $\substructure_2$.  The event
$\event_1=\calP(\basearcP_1) \cap \calP(\beta_{\match(1)}) \neq \emptyset)$
is always false since $\beta_{\match(1)}$ does not exist.  So all clause
containing $\event_1$ is false.  To make each of the remaining clauses true, $\event_2$ must be
true, and the case reduces to the previous case (by simply omitting node $\substructure_1$).  So the
same argument again renders a contradiction.  This completes the proof of the theorem.

\end{proof}

\section{Constructing $\calH$}
\label{sec:calH}

In this section, we describe how to construct the set $\calH$ in details.
We first include in $\calH$
all square gadgets.
The boundary of $\calH$ consists of several closed curves, as shown in
Figure~\ref{fig:04:03}.
$\calH$ and all uncovered arcs define a set of substructures.

First, we note that there may exist a closed curve that all points on the curve are covered by some
arcs (or informally, we have a cyclic substructure, with the baseline being a cycle).
We need to break all such baseline cycles
by including a constant number of
extra arcs into $\calH$.
This is easy after we introduce the label-cut operation in Section~\ref{subsec:merge},
and we will spell out all details then.
Note that we cannot choose some arbitrary envelope cycle
since it may ruin some good properties we want to maintain.

\eat{
If there is a point on the baseline which cannot be covered by any arc in an active region, we
consider the set of arcs which can cover the point.
By including an envelope arc in this set, the baseline cycle can be broken.
If there does not exist such a point,
we know there exist at least two active regions along the the baseline
and they are associated with two non-adjacent squares.
Consider  one of such active regions.
The baseline covered by all arcs of this active region is only a segment
(just part of the closed curve).
Suppose the endpoints of the segment are $A$ and $B$.
We select one point on the closed curve which is not covered by the active region.
Regard it as a start point and select a orientation (such as clockwise),
then we can do cut operations at $A$ and $B$ respectively.
(More concretely, it is label-cut operation. See Definition~\ref{def:labelcutoff}). Then the
closed curve is divided into two part and each part has two endpoints.
}

From now on, we assume that all baselines are simple paths.
Now, each closed curve contains one or more baselines. So, we have an initial set
of well defined substructures. The main purpose of this
section is to cut these initial substructures such that Lemma~\ref{lm:calH} holds.

We will execute a series of operations for constructing $\calH$.
We first provide below a high level sketch of our algorithm,
and outline how the substructures and the substructure relation graph $\auxgraph$
evolve along with the operations.

\begin{itemize}
\item
(Section~\ref{subsec:merge})
First, we deal with active regions.
Sometimes, two active region may overlap significantly and become inseparable
(formally defined later), they essentially need to be dealt as a single active region.
In this case, we merge the two active regions together (we do not need to do anything,
but just to pretend that there is only one active region).
We can also show that one active region can be merged with at most one other active region.
For the rest of cases, two overlapping active region are separable,
and we can cut them into at most two non-overlapping active regions, by adding
a small number of extra disks in $\calH$.
After the merging and cutting operations, each substructure
contains at most one active region.
Hence, the substructures satisfy the property (P1) in Lemma~\ref{lm:calH}.
Moreover, we show that if any substructure contains an active
region, the substructure is limited in a small region.

\item  (Section~\ref{subsec:cutoff})
We ensure that each substructure is non-self-intersecting by a simple greedy algorithm. After this step,
(P2) is satisfied.

\item (Section~\ref{subsec:GOC})
  In this step, we ensure that substructure relation
  graph $\auxgraph$ is a acyclic 2-matching (P3).
  The step has three stages.
  First, we prove that the set of blue edges forms a matching.
  Second, we give an algorithm for cutting the substructures which overlap with
  two or more other substructures.
  After the cut, each substructure overlaps with no more than one
  other substructure.
  So after the first two stages, we can see
  that $\auxgraph$ is composed of a blue matching and a red matching.
  At last, we prove that the blue edges and red edges cannot form
  a cycle, establishing $\auxgraph$ is acyclic.

  \item (Section~\ref{subsec:POC})
  The goal of this step is to ensure the point-order consistency (P4).
  We first show there does not exist
  a point covered by more than two substructures, when $\auxgraph$
  is an acyclic 2-matching.
  Hence, we only need to handle the case of two overlapping substructures.
  We show it is enough to break all cycles in a certain planar directed graph. Again, we can add a few more disks to cut all such cycles. 	
  \item (Section~\ref{subsec:number})
  Lastly, we show that the number of disks added in $\calH$ in the above four
  steps is $O(K^2)$, where $K = \frac{L}{\mu}$ and $L$ and $\mu$ are side lengths of block and small
  square respectively.
\end{itemize}

\subsection{Merging  and Cutting  Active Regions}
\label{subsec:merge}

If two active regions overlap
in the same substructure, we need to either merge them into a new
one or cut them into two non-overlapping new ones.
As we know, each gadget may have two active
regions.
Suppose active regions $\activeR_1$ and $\activeR_2$ belong to
the same gadget $\gadget$,
while
$\activeR_1'$ and $\activeR_2'$ belong to a different gadget $\gadget'$.
Due to R-correlations, we need consider the four active
regions together.

First, we consider the case where $\activeR_1$ overlaps with $\activeR_1'$,
and $\activeR_2$ overlaps with $\activeR_2'$.
We need the following important concept \emph{order-separability},
which characterizes how the two sets of arcs overlap.

\begin{definition}[Order-separability]
  \label{def:order separable}
 Consider a substructure $\substructure(\baseline, \calA)$. $\calA_1, \calA_2$ are two disjoint
 subsets of
  $\calA$. If $\calA_1, \calA_2$ satisfy that
    \begin{equation}
      \label{eq:04:02}
      a \prec b  , \text{ for any } a \in \calA_1 \text{ and } b \in \calA_2,
    \end{equation}
    we say that $\calA_1, \calA_2$ are order-separable.
\end{definition}

We use $\calA_1$ , $\calA_1'$, $\calA_2$, $\calA_2'$
to denote the set of arc associated with
active regions $\activeR_1$, $\activeR_1'$, $\activeR_2$, $\activeR_2'$, respectively.
If $\calA_1$ and $\calA'_1$ are not order-separable, we say the pair $(\activeR_1 ,
\activeR_1')$ is a \emph{mixture}.   If both of $(\activeR_1 ,\activeR_1')$ and
$(\activeR_2, \activeR_2')$ are mixtures, we say the two pairs form a
\emph{double-mixture}.
When they are double-mixture, we merge them simultaneously.
It only means that we regard the two active regions $(\activeR_1 ,\activeR_1')$ as a new single active region, and $(\activeR_2, \activeR_2')$ as another single active region.

To show an active region cannot grow unbounded, we prove
that the merge operations do not generate chain reactions.
The rough idea is that if two active regions form a mixture,
their corresponding small squares must be adjacent and their  core-central areas must overlap.
Due to the special shape of core-central areas (a narrow spindle shape), the overlapping can only happen
between two of them, not more.

\begin{lemma}
\label{double-mix}
Consider two non-empty small squares $\ssquare$, $\ssquare'$.
Suppose the square gadgets in the two
squares are $\gadget(\ssquare)=(\Disk_s, \Disk_t)$
and $\gadget'(\ssquare')=(\Disk_s', \Disk_t')$.
The active region pairs
$(\activeR_1, \activeR_2)$ and $( \activeR_1', \activeR_2' )$ are associated with
gadget $\gadget(\ssquare)$
and $\gadget'(\ssquare')$ respectively.
$(\activeR_1, \activeR_2)$ and $(\activeR_1',
\activeR_2')$ form a double-mixture.
Then, the following statements hold:
\begin{enumerate}
\item Their corresponding squares $\ssquare, \ssquare'$ are adjacent;
\item The core-central areas of $\gadget(\ssquare)$ and $\gadget'(\ssquare')$ overlap;
\item The angle between $\dcenter_s\dcenter_t$ and $\dcenter_s'\dcenter_t'$ is $O(\epsilon)$
\item None of the two core-central areas can overlap
with any small squares other than $\ssquare$ and $\ssquare'$.
\end{enumerate}
\end{lemma}

Next, we consider the case where
only one of $(\activeR_1, \activeR_1')$ and $(\activeR_2, \activeR_2')$ is a mixture.
However, we show that it is impossible as follows.
We use notation $\activeR(\calA)$ to denote an active region with arc set $\calA$.

\begin{lemma}
\label{lm:no single mixture}
Suppose active region pairs $(\activeR_1(\calA_1)$, $ \activeR_2(\calA_2) )$ and $(
\activeR_1'( \calA_1')$, $\activeR_2'(\calA_2'))$ are associated with gadget $\gadget$ and
$\gadget'$ respectively.
If $\calA_1$ and  $\calA_1'$ are
order-separable, then  $\calA_2$ and $\calA_2'$  are also order-separable.
\end{lemma}

The proofs of Lemma~\ref{double-mix} and Lemma~\ref{lm:no single mixture}
are elementary planar geometry and we defer them to Appendix~\ref{apd:merge}.

\topic{Cutting Overlapping Active Regions}
After the merging stage,
even if any two active regions overlap in the same substructure,
their arcs are order-separable.
We define \emph{label-cut} operation as below
to further separate them such that the baselines of all the active regions
are non-overlapping.

We consider two overlapping active regions
$\activeR_1(\calA_1)$ and $\activeR_2(\calA_2)$ in a
substructure $\substructure(\calA)$.
Suppose $\calA_1, \calA_2 \subset \calA$ and $\calA_1,
\calA_2$ are order-separable.
We can cut  $\substructure$ into two substructures $\substructure_1$ and
$\substructure_2$ with disjoint baselines,
such that $\Region(\activeR_1)  \subset \Region(\substructure_1) $ and
$\Region(\activeR_2) \subset \Region(\substructure_2)  $.
In other words,  if we assign
arcs in $\calA_1$ one kind of label
and arcs in $\calA_2$ a different kind of label,
after the label-cut operation, the arcs with different labels belong
to different new substructures.
We define label-cut formally as follows.

\begin{figure}[t]
	\centering
	\includegraphics[width=0.9\textwidth]{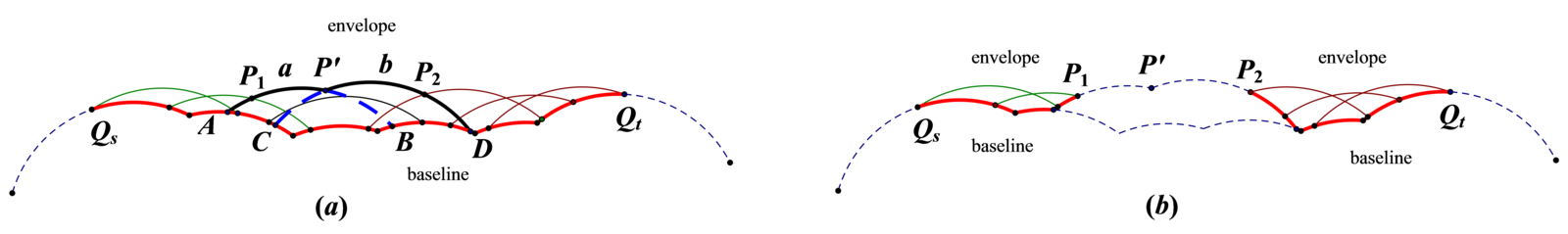}
	\caption{The example of label-cut. The left hand side  illustrates the whole substructure before
    cutting. The arcs have two different labels. One is green and the other is brown. The bold black
  subarcs are what we select in the envelope. The the right hand side  illustrates that each of the
  two separable substructures induced by the label-cut operation only contains arcs with the same label.}
	\label{fig:labelcutoff}
\end{figure}

\begin{definition}[Label-cut]
    \label{def:labelcutoff}
    Consider a substructure $\substructure(\baseline, \calA)$. Suppose $\calA_1', \calA_2'$ are two
    subsets of $\calA$. $\calA_1'$ and $\calA_2'$ are order-separable.
    We can add two consecutive arcs in the envelope of $\calA$ into $\calH$.
    After that, $\substructure$ can be
    separated into two new substructures  $\substructure_1(\baseline_1, \calA_1),
    \substructure_2(\baseline_2, \calA_2)$,
    such that subarcs in $\calA_1'$ only belong to $\calA_1$, and
    subarcs in $\calA_2'$ only belong to $\calA_2$.
\end{definition}

\begin{lemma}
  \label{lm:labelcut}
Consider a substructure $\substructure(\baseline, \calA)$ and two subsets $\calA_1', \calA_2'$ of
$\calA$. There exists a label-cut when $\calA_1'$ and $\calA_2'$ are order-separable, and
we can perform the label-cut in polynomial time.
\end{lemma}
Figure~\ref{fig:labelcutoff} illustrates the process of construction.  We defer its proof in
Appendix~\ref{apd:merge}.


After the label-cut operation, in each substructure,
the baselines for all active regions are not
overlapping.
Thus, if any substructure contains more than one active region,
consider any two of them, say $\calA_1$ and $\calA_2$.
We can add into $\calH$ one arc $a$ along the envelope
which satisfies $a_1 \prec a \prec a_2, \forall
a_1 \in \calA_1, a_2\in \calA_2$.
After the addition of arc $a$, $\calA_1$ and $\calA_2$ are
separated into two different new substructures.
Repeat the above step whenever one substructure contains more than
one active region.
This establishes the active region uniqueness property (P1).

\topic{Limiting the size of substructure which contains an active region}
Now, we discuss how to make substructure which contains an active region
bounded inside a small region.
This property is particular useful later
when we show the substructure relation graph $\auxgraph$ is acyclic.


Suppose the gadget of square $\ssquare$ is $(\Disk_s, \Disk_t)$.
The line
$\dcenter_s\dcenter_t$ divides the plane into two halfplanes $H^{+}$ and
$H^{-}$.
$\Disk_s$ and $\Disk_t$ intersect at point $P$ and $Q$.
The boundary of disk $\Disk(P, 2)$
is tangent to  $\Disk_s$ and $\Disk_t$ at point $Q_s$ and $Q_t$ respectively.
$\Disk(Q_s,1)$ and
$\Disk(Q_t,1)$ intersect at point $D$. See Figure~\ref{fig:smallActiveRegion}.
We call $D$ the dome-point of gadget $\gadget(\ssquare)$ and
use $\Domain(\ssquare^+)$  to denote the
region $(\Disk(D,1)) - \Disk_s - \Disk_t )\cap H^{+}$.

\begin{figure}[t]
  \centering
  \includegraphics[width=0.35\textwidth]{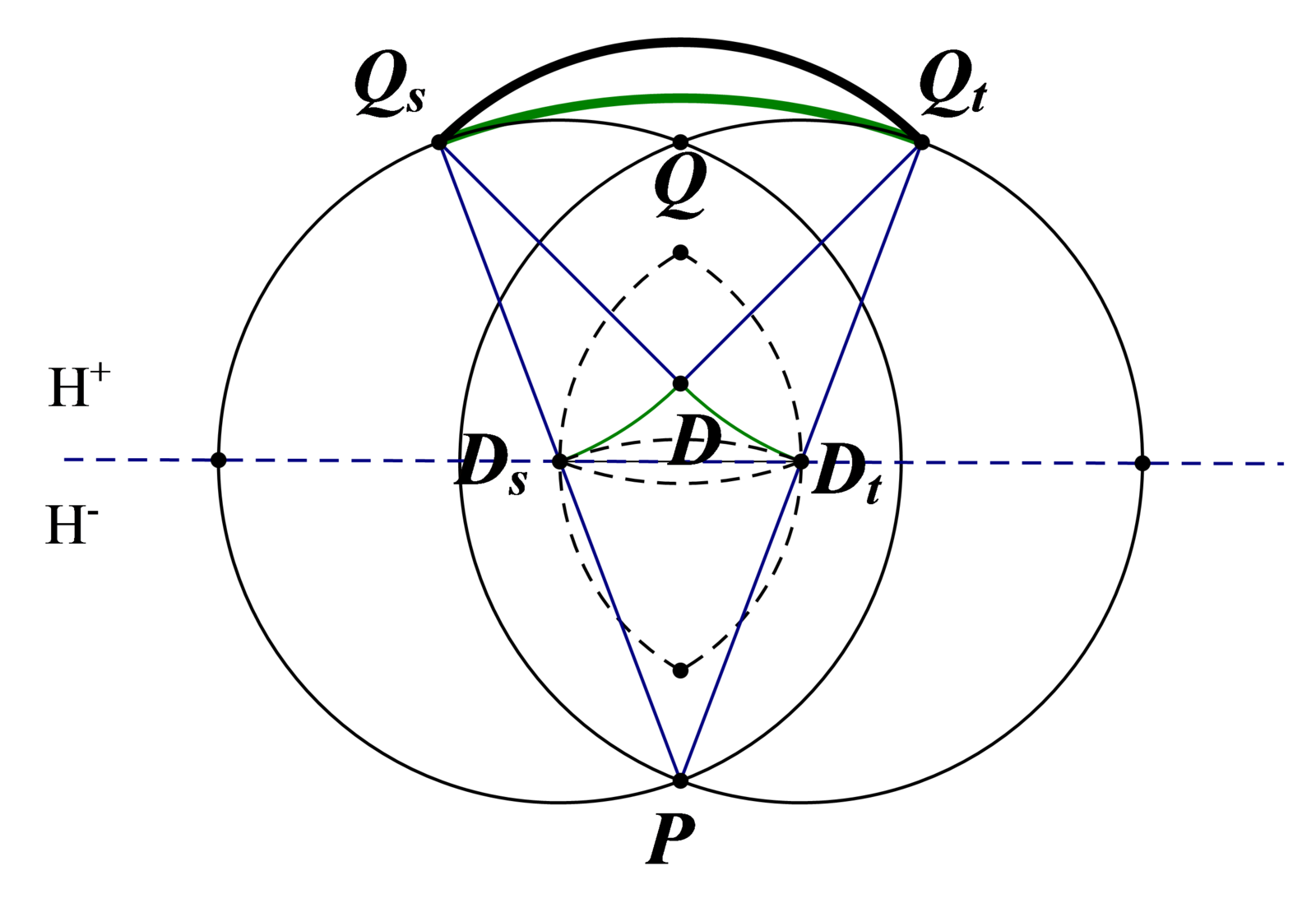}
  \caption{The farthest disk pair of square $\ssquare$ is $(\Disk_s, \Disk_t)$. Suppose  $\Disk_s$
    and $\Disk_t$ intersect at point $P, Q$. The disk $\Disk(P, 2)$
is tangent to $\Disk_s$ and $\Disk_t$ at point $Q_s$ and  $Q_t$ respectively. $\Disk(Q_s,1)$ and
$\Disk(Q_t,1)$ intersect at $D$. The active region of $\ssquare$ in $H^{+}$ is totally covered by $\Domain(\ssquare^+)$. }
  \label{fig:smallActiveRegion}
\end{figure}

\begin{lemma}
  \label{smallregion}
  Consider the substructure $\substructure$ which contains an active region of square
  $\ssquare$.  The substructure can be  cut  into at most three smaller substructures, by
  doing label-cut twice.
  At most one of them
  contains the active region. Moreover,  this new substructure (if any) is
  bounded by the region $\Domain(\ssquare^+)$ (or $\Domain(\ssquare^{-})$) associated with $\ssquare$.
\end{lemma}

We defer the proof to Appendix~\ref{apd:merge}.

There may exist  a substructure which contains a merged active region (i.e. an active region
which is the union of two initial active regions).
Based on Lemma~\ref{double-mix}, the arcs of the two initial active regions belong to two adjacent small
squares. Suppose the two squares are  $\ssquare$, $\ssquare'$, with square gadgets
($\Disk_s, \Disk_t$) and ($\Disk_s', \Disk_t'$) respectively.
The dome-points of $\ssquare$ and
$\ssquare'$ in $H^+$ are  respectively $\dcenter$ and $\dcenter'$.
We apply the operations in
Lemma~\ref{smallregion} for each
initial active region.
Since the angle between $\dcenter_s\dcenter_t$ and $\dcenter_s'\dcenter_t'$
is $O(\epsilon)$,  obviously,  the substructure containing the merged active region is small as
well, i.e. bounded in $\Domain(\ssquare^+) \cup \Domain(\ssquare'^+)$.

Since each substructure which contains an active region is small enough,
we can show the following lemma, which will be useful for proving the
acyclicity of the substructure relation graph $\auxgraph$.

\begin{lemma}[Highly parallel arcs]
  \label{smallangle}
  Suppose substructure $\substructure(\baseline, \calA)$ contains an active region. The central
  angle of any arc in $\substructure$ is no more than $O(\epsilon)$.
  Meanwhile, there exists a
  line $l$ such that the angle between $l$ and the tangent line at any point of
  any arc $a \in \calA$ is at most $O(\epsilon)$.
\end{lemma}

Actually, for the initial active region, the line $\dcenter_s\dcenter_t$ satisfy the property. For
the merged active region, we know the angle between $\dcenter_s\dcenter_t$ and
$\dcenter_s'\dcenter_t'$ is $O(\epsilon)$ based on Lemma~\ref{double-mix}.
Hence, we can still see that the line
$\dcenter_s\dcenter_t$ satisfies the property.

To summarize, after all operations in this subsection,
we can ensure that each substructure contains at most one active region
(i.e., (P1) in Lemma~\ref{lm:calH}).
Moreover, we have that
each substructure which contains an active region is small enough
(so that Lemma~\ref{smallangle} holds).

\topic{Handling cyclic substructures}
At the end of this subsection,
we deal with the problem we left in the very beginning of Section~\ref{sec:calH},
to break all cyclic baselines.
Note that this step should be done in the beginning.
First,
we consider that case that there exists a point
on the baseline which cannot be covered by any
arc of any active region.
We can include any envelope arc that covers the point into $\calH$,
which is enough to break the cycle.
This is essentially a label-cut and do not separate any single connected active region
into disconnected pieces.
Then we consider the
case where every point on the baseline
is covered by some arc of active regions.
Note that the merge operation only depends on the local property
of two active regions, thus, we can merge active regions even in ``cyclic substructure''.
Assume that we have done all merge operations.
Based on
Lemma~\ref{smallangle},
we know one active
region is very small and thus cannot cover all points on a closed curve.
We pick one active region.
Using the same operation as Lemma~\ref{smallregion} (do two label-cuts),
we can essentially isolate the active region and
cut the original cyclic baseline to two new baselines.

\subsection{Eliminating Self-intersections}
\label{subsec:cutoff}

The goal of this part is to add a few more disks into $\calH$ so
that any substructure is non-self-intersecting (Recall the definition in Lemma~\ref{lemma:singleIntersection}).
Note that the substructures which contain active region are non-self-intersecting after the process in Section~\ref{subsec:merge}. We just need to process the substructures without active regions in this part.

\begin{figure}[t]
	\centering
	\includegraphics[width=0.65\textwidth]{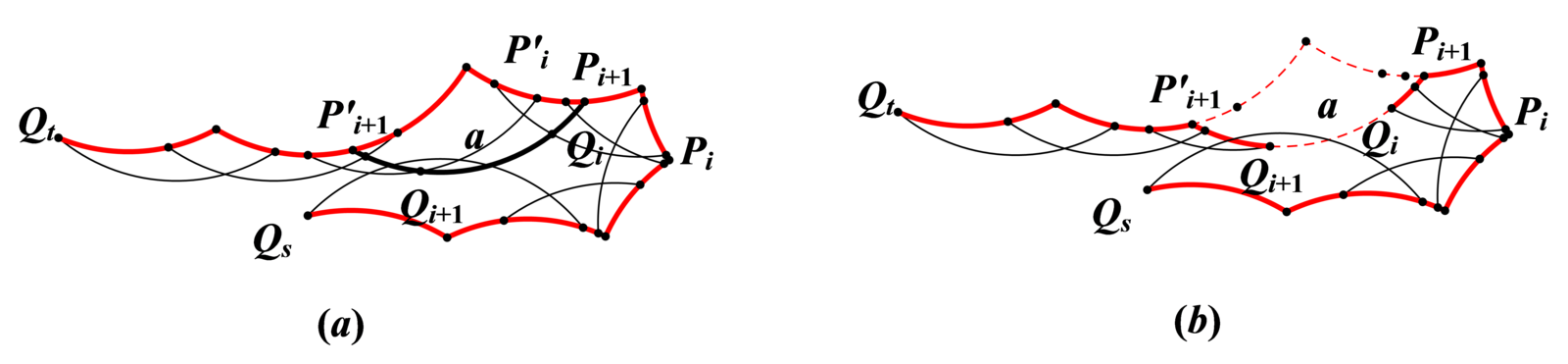}
	\caption{The process to avoid self-intersection. The left hand side is a self-intersection
    substructure. We search from point $Q_s$ along the envelope. Let arc set
$\calA_i$ be $\{a_1 ,a_2 , \ldots, a_i  \}$. Then we have a set $\{ \substructure_i[\calA_{i}]
\}_{i\in[k]}$ of substructures. If $\substructure_{i}$ is non-self-intersecting but $\substructure_{i+1}$
is  self-intersecting.  We add arc $a_{i+1}$ in $\calH$. The right hand side illustrates the two new
substructures after the cut.}
	\label{fig:point-cut-off}
\end{figure}

We use a simple greedy approach.
We consider one substructure $\substructure(\baseline, \calA)$.
Suppose the endpoints of
$\baseline$ are $Q_s$ and $Q_t$, and the envelope is
$\{a_1[Q_s, Q_1]$,$\ldots$, $a_k[Q_{k-1}, Q_t] \}$,
where $Q_i$ is
the intersection point of $a_{i-1}$ and $a_i$.
We denote the
endpoints of $a_i$ on the baseline $\baseline$ by $P_i$ and  $P_i'$ (Note that $P_1 = Q_s$).
Let arc set
$\calA_i$ be $\{a_1 ,a_2 , \ldots, a_i  \}$.
Then we have a set $\{\substructure_i[\calA_{i}]\}_{i\in[k]}$
of substructures, where $\substructure_i[\calA_{i}]$ is the substructure induced by
arc set $\calA_{i}$.
We consider the arcs lying on the envelope one by
one and check whether we should add it into $\calH$ or not.
Concretely, we add $\Disk(a_{i+1})$ in $\calH$ if the following condition holds:
\begin{itemize}
\item $\substructure_i$ is non-self-intersecting, but $\substructure_{i+1}$ is self-intersecting,
\end{itemize}
The addition of $\Disk(a_{i+1})$ cuts the substructure into two smaller substructures.
One of which is certainly non-self-intersecting. The
other is induced by arc set $\calA = \{a_{i+2}, \ldots, a_{k} \}$.
See Figure~\ref{fig:point-cut-off}.
We repeat
the above process until there is no self-intersection in all substructures.
Furthermore, we can easily prove the following nice property.
The proof can be found in Appendix~\ref{apd:cutoff}.

\begin{lemma}
  \label{arcuniqueness}
  In each of the above iterations, one substructure $\substructure(\baseline, \calA)$ is cut into
  at most two new substructures. Any original arc in
  $\calA$ cannot be cut into two pieces, each of which belongs to a different new substructure.
\end{lemma}
\eat{
\begin{proof}
Suppose the envelope is $\{a_1[Q_s, Q_1], \ldots, a_k[Q_{k-1},Q_t] \}$.
After we add the disk $\Disk(a_i)$ in
$\calH$, the baselines of the two new substructures are
$\baseline_1 = \baseline[Q_s, P_{i+1}] \cup
a_{i+1}[P_{i+1}, Q_i]$ and
$\baseline_2 = a_{i+1}[Q_{i+1},P_{i+1}'] \cup \baseline[P_{i+1}', Q_t]$.
Since $a_i$ is an arc lying on the envelope, there does not exist an arc $b$ with endpoints
$(A, B)$ such that $A \prec P_i \prec P_{i+1} \prec B$.
So the endpoints of any arc in
$\substructure_1$ cannot locate in $\Region(\substructure_2)$. Thus, any arc cannot be separated
into two substructures.
\end{proof}
}

To summarize, we have obtained the non-self-intersection property ((P2) in Lemma~\ref{lm:calH}).

\subsection{Ensuring that $\auxgraph$ is an Acyclic 2-Matching }
\label{subsec:GOC}

We discuss how to add some extra disks in $\calH$
to make $\auxgraph$ an acyclic 2-matching ((P3) in Lemma~\ref{lm:calH}).

\topic{Blue edges}
First we show that the set of blue edges form a matching.
\begin{lemma}
  Two blue edges cannot be incident to the same node.
\end{lemma}
\begin{proof}
  Before the merge operation, the set of active region pairs forms a matching.
  To see this, note that our merge operations always apply to a double-mixture
  (which corresponds to merging two blue edges into one).
  Moreover, any cut operation cannot
break one active region into two, thus has no effect on any blue edge.
Hence, after all merge and cut operations, the
set of active region pairs is still  a matching.
\end{proof}

\topic{Red edges}
Then, we prove that any node which has more than one incident red edges can be cut
such that each new node(i.e., substructure) only has at most one incident red edge.

First, we prove a simple yet useful geometric lemma stating
that a point cannot be covered by three or more substructures.
Note that from now on, all substructures have no self-intersections.

\begin{lemma}
  \label{oneByone}
  We are given a substructure $\substructure(\baseline,\calA)$.
  $a \in \calA$ and two arcs $b_1, b_2 \notin \calA$.
  If $b_1, b_2$ cover the same point on
  $a$,  $b_1, b_2$ should belong to the same substructure.
\end{lemma}

Intuitively, if the two disks corresponding to $b_1$ and $b_2$
cover the same point, they should be close enough such that
their corresponding square gadgets overlap (which implies $b_1$ and $b_2$ share the same baseline).
First we prove that the minimum distance between any two disks in two different
substructure should not be too small, i.e., their overlapping region cannot be too large.
The proof can be found in Appendix~\ref{apd:GOC}.


In fact, essentially the same proof can be used to
prove that
any two different substructures cannot both intersect with subarc
whose central angle is $O(\epsilon)$, as in the following corollary.
\begin{corollary}
  \label{arcByarc}
   Suppose $\substructure(\baseline,\calA)$ is a substructure without any self-intersection.
   Consider an arc  $a \in \calA$ and two arcs $b_1, b_2 \notin \calA$.
   Suppose $a'$ is a subarc of $a$ whose central
   angle is $O(\epsilon)$.
   If both $b_1, b_2$ cover some part of $a'$, $b_1, b_2$
  should belong to the same substructure.
\end{corollary}

Combining with
Lemma~\ref{smallangle}, we can easily see the following lemma.

\begin{lemma}
  \label{NoAT}
  Any substructure which contains an active region cannot overlap with two or more different substructures.
\end{lemma}

Then we show how to cut the substructure $\substructure(\baseline, \calA)$ which overlaps with more than
one other substructures.
Note that such substructure does not contain an active region based on Lemma~\ref{NoAT}.
Suppose the envelope of $\substructure$ is
$\{a_1[Q_s, Q_1], \ldots, a_k[Q_{k-1},Q_t] \}$.
$\substructure$ overlaps with $k$
substructures $\substructure_i(\baseline_i,\calA_i), i = 1,2, \ldots, k$.

If $\substructure_i$ overlaps with $\substructure$,
there exists an arc $a \in \substructure_i$ intersecting
some envelope arc of $\substructure$.
Thus, the envelope $\substructure$ is subdivided into several segments by
those intersection points. We can label those segments as follows:
\begin{itemize}
\item If the segment is covered by some arc in $\substructure_i$, we label it as `$i$'.
\item If there is no arc in any $\substructure_i$ covering the segment, we label it as  `$0$'.
\end{itemize}
See Figure~\ref{fig:arclabel} for an example.
According to the Lemma~\ref{oneByone}, we
know there is no point on the envelope covered by two substructures.
Thus the above labeling scheme is well defined.

\begin{figure}[t]
  \centering
  \includegraphics[width=0.45\textwidth]{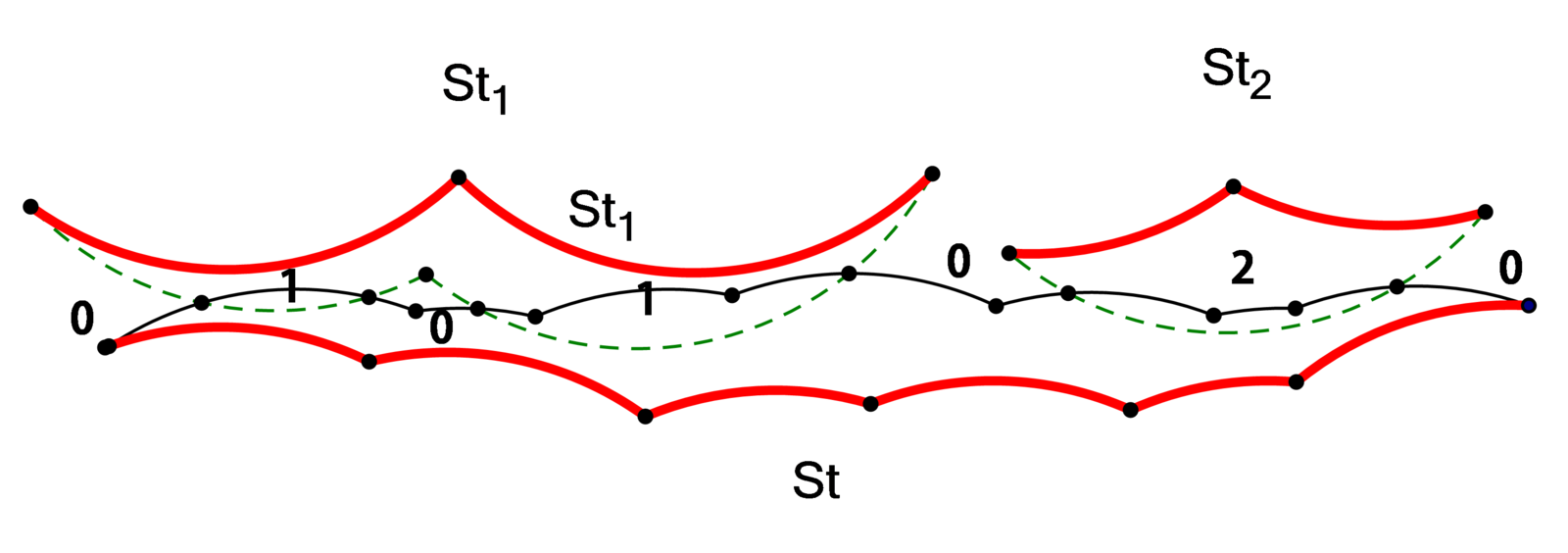}
  \caption{The arcs in substructures $\substructure_1$ and $\substructure_2$ cut  the envelope of
  $\substructure$ into 7 segments. The sequence of the labels for those segments is $0101020$.
  The compressed label sequence is $01020$. So we have $5$ l-segments.}
  \label{fig:arclabel}
\end{figure}

Traversing those segments along the envelope, we obtain a label sequence.
First, for each label $i$,
we identify those maximal consecutive subsequence, which consists of
only letter $0$ and $i$, and starts with and end with $i$,
and replace the subsequence by a single letter $i$.
We obtain a compressed sequence.
In fact, each letter, say $i$ ($i\ne 0$), in the compressed sequence corresponds
to one or more segments labeled with either $i$ or $0$, and
the first and last one must be labeled with $i$. 
We call the concatenation of those segments an l(ong)-segment.
Of course, a letter 0 in the compressed sequence corresponding to a segment with label 0.
Actually, the sequence is Davenport-Schinzel
sequence of order 2 \cite{Davenport},
because two different labels cannot intersect (because the
baselines of two substructures cannot intersect).
For example, the pattern ``$1212$'' should never appear.
Thus, the length of the compressed sequence, i.e.,
the number of l-segments, is at most $O(k)$.

Now, we discuss how to cut $\substructure$
into several new ones based on l-segments.
Keep in mind that our goal is to make sure
each new substructure only
overlap with one substructure of $\{\substructure_i\}_{i \in [k]}$.
The cut operation is again a simple greedy procedure.
Consider two consecutive l-segments.
Suppose they are $\{a_i[Q_{i-1}, Q_i]$, $a_{i+1}[Q_{i}, Q_{i+1}]$, $\ldots, a_j[Q_{j-1}, Q_j]\}$
and $\{a_{j +1}[Q_j, Q_{j+1}]$, $\ldots$, $a_k[Q_{k-1}, Q_k]\}$.
We add into $\calH$
the last arc $a_j$ of former
l-segment and the first arc $a_{j+1}$ of the later l-segment.
$\substructure$ is thus cut into two new ones.
Repeat the above step for all two consecutive l-segments in order.

We still need to show that after the cut, every new substructure
overlap at most one of $\{\substructure_i\}_{i\in [k]}$.
Consider one new substructure.
Notice such an original arc in $\calA$ can only belong to one new substructure.
By our cut operation, all envelope arcs
of the new substructure can intersect at most one of $\{ \substructure_i \}_{i \in [k]}$.
Hence, the new substructure can overlap at most one $\substructure_i$.


\eat{
\begin{lemma}
  If we ignore all ``$0$'' arcs, the minimum number of consecutive segments, in each of which the arcs contain
same label, is constant.
\end{lemma}
\begin{proof}
If we ignore the ``$0$'', we get a label
sequence such as ``$iiijjjiiikkkii\ldots$''. When the number of segments is minimum, the consecutive arcs
with same label must belong to the same segment. Note that there may exist two separate segments with
the same label ``$i$''. But two different labels cannot intersect because the baselines of two
substructures cannot intersect. For example, the sequence such as ``$ijjiijj$'' cannot
happen. So, in the sequence, there exists at least one label ``$i$'' which is consecutive. They
belong to the same segment. If we delete them, there must exist another label ``$i'$'' which is
consecutive. The delete operation is able to merge at most two separable segments. So do the
operation in turn, we can prove there are at most $2k-1$ segments.
\end{proof}
}

\topic{Blue edges and red edges}
After the above operations,
the set of all blue edges forms a matching,
while the set of all red edges
also forms a matching.
To show $\auxgraph$ is an acyclic 2-matching,
it suffices to prove $\auxgraph$ contains no cycle.
So, the rest of the section is devoted to prove the following lemma.

\begin{lemma}
  \label{lm:Nocycle}
  Suppose the side length of square is $\mu$, where $\mu = O(\epsilon)$ and
  the block contains $K \times K$ small squares, where $K = \frac{C_0}{\epsilon^2}$
  and $C_0$ is an appropriate constant.
  Then, after all operations stated in this section,
  there is no cycle in $\auxgraph$.
\end{lemma}

If there is a cycle $\cycle$ in $\auxgraph$,
the red edges and blue edges alternate in $\cycle$, which
correspond to a sequence of substructures, each
containing an active region (since it is R-correlated with another active region).
Now, we provide a high level explanation why
Lemma~\ref{lm:Nocycle} should hold.
Each active region is associated with a small square.
A small square is very small (comparing to a unit disk or the whole block),
and an active region is also very small.
We pick a point in each small square
and each substructure.
If two substructures in $\cycle$ overlap,  the two points in them are also
very close (i.e., $O(\epsilon)$). So we can pretend the two points as one point
(or we just pick a point in their overlapping region).
For each active region $\activeR$, we connect the point in
$\activeR$ and the point in the small square
associated with $\activeR$.
Since Both the small square and the active region are very small,
the distance between two point is about $1 - O(\epsilon)$.
Thus, a cycle $\cycle$ would present itself geometrically as a polygon.
We can show the angle between two adjacent
edges of the polygon is close to $\pi$.
So the size of the polygon cannot be not small (it takes
a lot of edges to wrap a loop).
However, the polygon cannot be much larger than the block.
By contradiction, we prove that there is no cycle in $\auxgraph$.

Now, we formally prove Lemma~\ref{lm:Nocycle}.
We first prove a geometric lemma which will be useful for bounding the angle
between the two adjacent edges of the aforementioned polygon.

\begin{lemma}
  \label{interangle}
  Consider two substructures $\substructure_1(\baseline_1, \calA_1)$ and
  $\substructure_2(\baseline_1, \calA_2)$ in $\cycle$.
  Suppose $\substructure_1$ and $\substructure_2$
  overlap (there is a red edge between them).
  For any two arcs $a \in \calA_1$ and $b \in
  \calA_2$,
  suppose that $\Disk(a)$ and $\Disk(b)$ overlap and
  their intersection points are $A$ and $B$.
  The central angles of $a, b$ are $\theta_a, \theta_b$ respectively.
  The centers of $\Disk(a), \Disk(b)$ are $\dcenter_a, \dcenter_b$.
  Then $\angle A\dcenter_a B$ (or $\angle
  A \dcenter_bB$) is  at most   $(\theta_a + \theta_b)$.
\end{lemma}

\begin{proof}
  We distinguish a few cases depending on
  whether the intersection points $A$ and $B$ lie on
  $a$ or $b$ or none of them.
  All cases are depicted in Figure~\ref{fig:interangle}.
\begin{itemize}
\item Both $A$ and $B$ lie on one arc (see Figure (a)(b)).
  W.l.o.g., suppose they lie on arc $a$.
  Obviously, $\angle A\dcenter_aB$ is no more than $\theta_a$ (or $\theta_b$) .
\item If one intersection is on neither $a$ nor $b$ (see Figure (c)),
  we prove that the case
  cannot happen. Suppose $A$ is on neither of $a$ and $b$.
  The endpoint $A_1$ of $a$ is covered by
  disk $\Disk(b)$ and the endpoint $A_2$ of $a$ is covered by disk $\Disk(a)$.
  Thus the baseline $\baseline_1$
  must intersect with $\baseline_2$ which contradicts the fact that $\substructure_1$ and
  $\substructure_2$ are two different substructures.
\item If one intersection point is on $a$ but not on $b$ and
  the other intersection point is on $b$ but not on
  $a$ (see Figure (d)), it is easy to see that $B_1$ is covered by arc $b$,
  and $A_2$ is covered by $a$ since their baselines do not intersect.
  Thus, the length of arc $AB$ (w.r.t. $\Disk(a)$)
  is no more than the sum of lengths of $a$ and $b$.
  So $\angle A\dcenter_a B$ is at most $\theta_a + \theta_b$.
  \end{itemize}
The above cases are exhaustive, thus our proof is completed.
\end{proof}

\begin{figure}[t]
  \centering
  \includegraphics[width=0.5\textwidth]{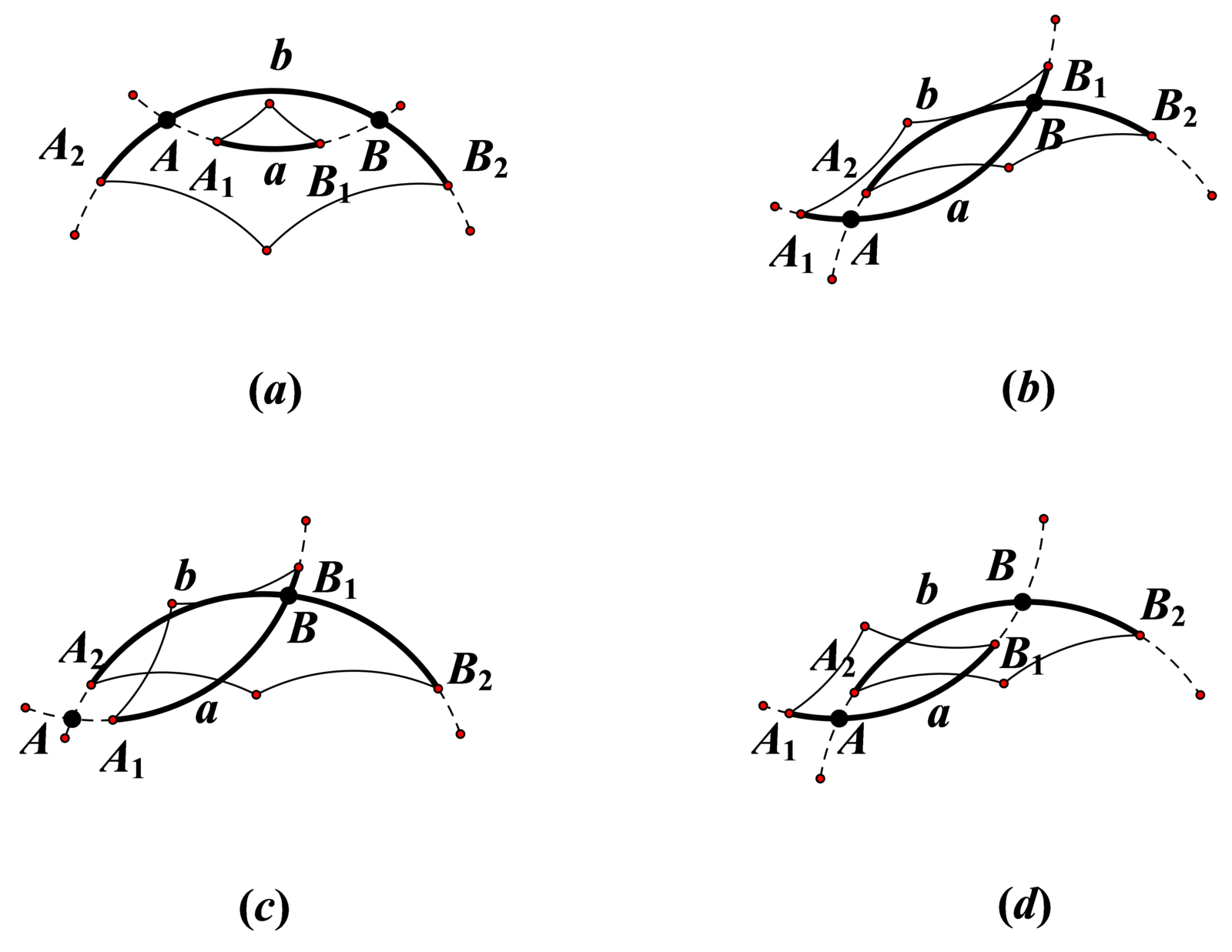}
  \caption{The four kinds of cases for two overlapping arcs.}
  \label{fig:interangle}
\end{figure}

Based on Lemma~\ref{smallangle},  Lemma~\ref{interangle}, we can prove Lemma~\ref{lm:Nocycle}
below:

\begin{proofoflm}{\ref{lm:Nocycle}}
Consider two substructures $\substructure_1$ and $\substructure_2$.
Suppose $\substructure_1$
contains the active region $\activeR_1$ ,
and $\substructure_2$ contains the active region
$\activeR_2$.
The centers of arcs of $\activeR_1$ and $\activeR_2$ locate in
small squares
$\ssquare_1$ and
$\ssquare_2$ respectively.
The square gadgets of the two
squares are $\gadget (\Disk_{s_1}, \Disk_{t_1})$ and $\gadget (\Disk_{s_2}, \Disk_{t_2})$ respectively.
If $\substructure_1$ and $\substructure_2$ overlap, their must exist an arc $a$ in
$\substructure_1$ and an arc $b$ in $\substructure_2$ such that $a$ and $b$ intersect.
Suppose $\Disk(a)$ and $\Disk(b)$ intersect at points $A$ and $B$.
The center of $\Disk(a)$ and $\Disk(b)$ are $\dcenter_a$
and $\dcenter_b$.

Based on Lemma~\ref{smallangle}, we know both central angle of $a$ and $b$ are no
more than $O(\epsilon)$.
According to Lemma~\ref{interangle},
the central angle of the arc $AB$ is at most $O(\epsilon)$.
It means angle between the tangent lines of $\Disk(a)$  and $\Disk(b)$ at point $A$
is no more than $O(\epsilon)$.
Thus, $\angle \dcenter_a A \dcenter_b$ is at least $\pi - O(\epsilon)$.

We know that all disks in the same active region are centered in one
small square or two adjacent small squares.
Moreover Lemma~\ref{smallangle} implies all disk centers should lies
in one or two squares.
Hence, the distance between (any point in) the square and (any point in) its active region
is at least $1-O(\epsilon)$.
Construct the aforementioned polygon.
Consider two adjacent edges $XY$ and $YZ$ in the polygon.
We consider two cases:
\begin{enumerate}
\item
$Y$ is in the intersection of
two substructures $\substructure_1$ and $\substructure_2$.
We can easily see that
(1) $|YA|=O(\epsilon)$;
(2) $|X\dcenter_a|=O(\epsilon)$;
(3) $|Z\dcenter_b|=O(\epsilon)$.
Hence, we can see $\angle XYZ$ is at least $\pi - O(\epsilon)$.
\item
$Y$ is in a small square $\ssquare$ and $X$ and $Z$ are in the the two substructures
associated with $\ssquare$. Since both substructures are bound in
an $O(\epsilon)$ size region (by Lemma~\ref{smallangle}), we can see that
$\angle XYZ$ is at least $\pi - O(\epsilon)$ as well.
\end{enumerate}
Hence, we can see the polygon
contains at least $\Omega(2\pi/\epsilon)$ nodes
and the diameter of the polygon is at least $\Omega(1/\epsilon)$.
But this cannot be larger than the diameter of a block, rendering a contradiction.

\hfill $\square$
\end{proofoflm}

To summarize,
we have ensured that $\auxgraph$ is an acyclic 2-matching ((P3) in Lemma~\ref{lm:calH}).

\subsection{Ensuring Point Order Consistence}
\label{subsec:POC}

We have ensured that the set of red edges is a matching in $\auxgraph$.
Hence, one substructure can overlap
with at most one other substructure.
Therefore, if we can guarantee that the points which are covered by any pair of overlapping
substructure are order consistent, then all points in $\calP$ are order consistence
(after all, point-order consistency is defined over a pair of substructures).

Consider two overlapping substructures
$\substructure_1(\baseline_1, \calA_1)$ and $\substructure_2(\baseline_2,\calA_2)$
and a set  $\calPco$ of points covered by $\calA_1 \cup \calA_2$.
Suppose $ P_1, P_2 \in \calPco$ and $a_1, a_2 \in \calA_1, b_1, b_2 \in \calA_2$.
Recall point-order consistency requires that the following conditions are met
\begin{itemize}
\item $P_1 \in \Region(a_1) \cap \Region(b_1) $ and $P_2 \in \Region(a_2) \cap \Region(b_2)$
\item  $P_1\notin \Region(a_2) \cup \Region(b_2)$ and $P_2 \notin \Region(a_1) \cup \Region(b_1)$.
\item  $(a_1 \prec a_2) \Leftrightarrow  (b_1 \prec b_2)$.
\end{itemize}

\begin{figure}[t]
	\centering
	\includegraphics[width=0.35\textwidth]{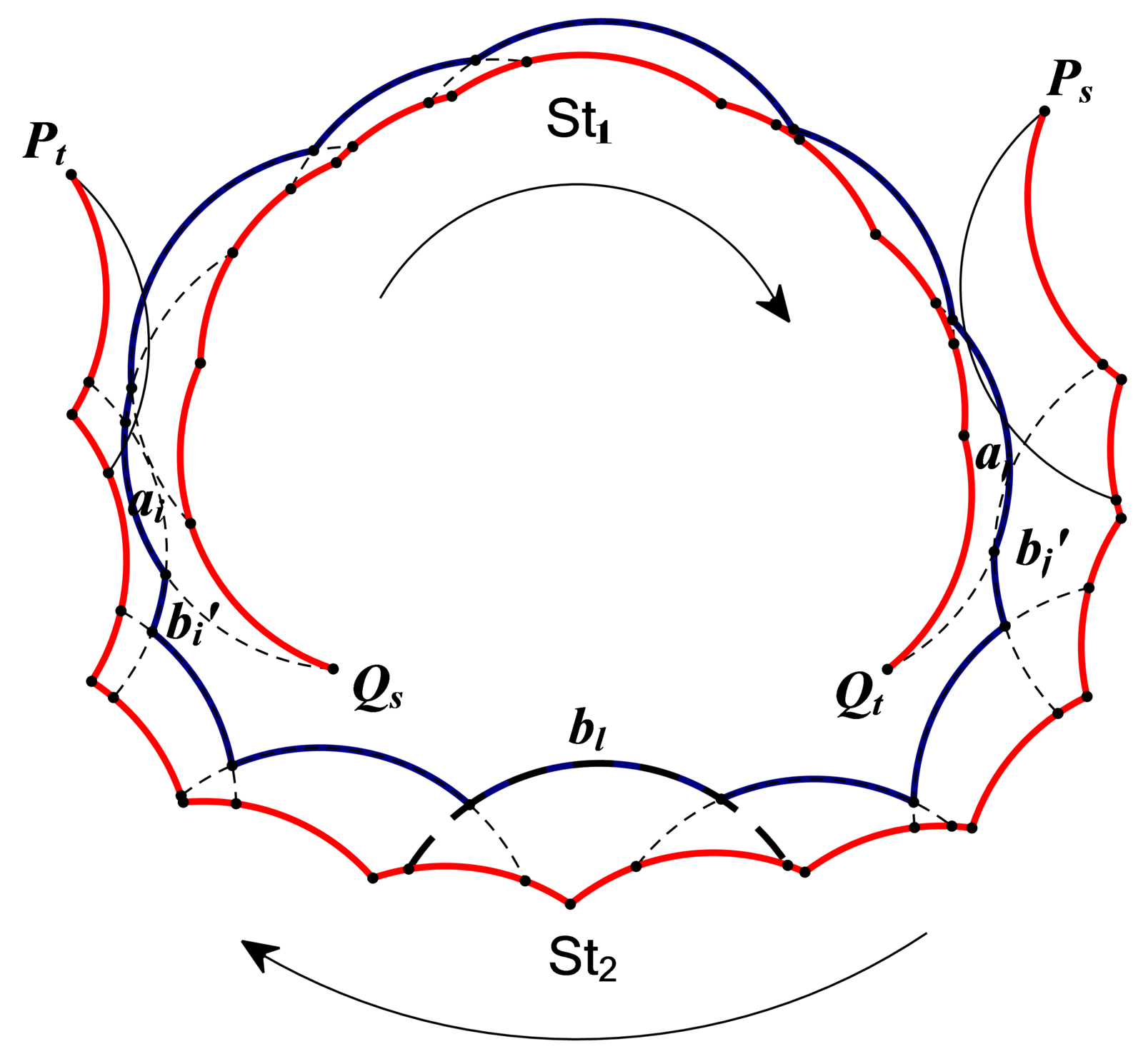}
	\caption{Substructure $\substructure_1(\baseline_1, \calA_1)$ and $\substructure_2(\baseline_2,
\calA_2)$ are overlapping.  $\baseline_1$ starts from $Q_s$ and ends up with
$Q_t$ and $\baseline_2$ starts from $P_s$ and ends up with $P_t$. There are two
paths forming a cycle.}
	\label{fig:pointorder}
\end{figure}

It is helpful to consider the following directed planar graph $\planargraph$
induced by all arcs in $\calA_1\cup\calA_2$ in the uncovered region $\Uncover$.
Regard each intersection point of arcs in $\calA_1 \cup \calA_2$ as a node.
Each subarc is an directed edge with a direction consistent with its baseline.
We use $A \to B$ to denote that there is a directed edge between nodes $A$ and $B$
in $\planargraph$.
If there is no directed cycle in $\planargraph$,
we can verify that all conditions listed above hold.
Indeed, suppose the condition is not satisfied,
which means $a_1 \prec a_2$, $b_1 \succ b_2$, $a_1, a_2 \in \calA_1$
and $b_1, b_2 \in \calA_2$. Suppose
$a_1, b_1$ intersect at $(A_1, B_1)$ and $A_1 \to B_1$, meanwhile $a_2, b_2$ intersect at $(A_2, B_2)$ and $A_2
\to B_2$.
Since $a_1 \prec a_2$,
there exists a path in $\substructure_1$ which goes from $B_1 $ to
$A_2$.
Similarly, there exists a path in $\substructure_2$ from $B_2$ to $A_1$. Thus, the
two paths and $a_1, a_2$ form a directed cycle.

So all we have to do is to break all cycles in $\planargraph$.
When $\planargraph$ contains a cycle,
we can cut the cycle through adding an arc on the envelope
into $\calH$.
See Figure~\ref{fig:pointorder} for an example.
Only arcs $\calA_1$ (or in $\calA_2$) cannot form a cycle.
So if there is a cycle, the cycle must pass through the envelope of $\calA_1$ and $\calA_2$.
Moreover, based on Lemma~\ref{smallangle}, if $\substructure_1$ and $\substructure_2$ form a cycle,
either $\substructure_1$ or $\substructure_2$ does not contains any active region. W.l.o.g., suppose
it is $\substructure_2$.  Based on these observations, we have our algorithms as follows:

Suppose the envelope of $\substructure_1$ is
$\Path_1 = \{a_1, a_2, \ldots, a_k\}$.
$a_i$ and $a_j$ are the first and last arcs respectively which intersect $\substructure_2$.
Suppose the arc $b_{i'} \in \calA_2$ overlaps with $a_i$,
(if there is more than one such arc,
we select a minimal one, w.r.t. the arc ordering)
and the arc $b_{j'} \in \calA_2$ overlaps
with $a_j$ (if there is more than one such arc, we select a maximal one).
Since $a_i \prec a_j$ and they do not satisfy point-order consistency,
we have $b_{i'} \succ b_{j'}$.
We can see that all arcs between $b_{i'}$ and $b_{j'}$ cannot intersect with $\Path_1$.
So we can select
one arc between $b_{i'}$ and $b_{j'}$ to add in $\calH$
for cutting  $\substructure_2$ into two.
After the cut, any cycle in $\planargraph$ can be broken.

After cutting  $\substructure_2$, $\substructure_1$  overlaps with both of the new substructures obtained from $\substructure_2$.
Then, we encounter the same situation as in
Section~\ref{subsec:GOC} (a node in $\auxgraph$ has
two incident red edges).
We can apply the operation in Section~\ref{subsec:GOC}
to cut $\substructure_1$ such that the set of red edges in $\auxgraph$ is still a matching.

\subsection{The number of disks in $\calH$}
\label{subsec:number}

Finally, we count collectively the total number of disks
that we have added in $\calH$.
First, we add the square gadget for each nonempty small square in $\calH$.
The number of the disks is $O(K^2)$, where $K=L/\mu=O(1/\epsilon^2)$.
In order to cut overlapping active regions,
the number of disks we add in $\calH$ is
bounded by the number of active regions.
Since there are $O(K^2)$ gadgets, we add $O(K^2)$ disks in
Section~\ref{subsec:merge}.
In Section~\ref{subsec:cutoff}, to ensure that each substructure is
non-self-intersecting and contains at most one active region,
we design a greedy algorithm.
We can see that
each arc we added in $\calH$ covers at least one intersection point of two disks in $\calH'$,
where $\calH'$ is the set $\calH$ before this step.
The algorithm guarantees that
each arc which we add in $\calH$ do not cover the same intersection point
on the boundary of $\calH'$.
Since the union complexity of unit disks is linear~\cite{Kedem86}
and  $\calH'$ contains at most $O(K^2)$ disks,
there are at most $O(K^2)$  intersection points on $\partial\calH'$.
So, we add at most $O(K^2)$ in this step.
In Section~\ref{subsec:GOC}, we break the cycles in $\auxgraph$.
The number of disks we add $\calH$ is  proportional to the number of substructures.
So, again, we add at most $O(K^2)$ disks.
Similarly, in Section~\ref{subsec:POC}, we also add at most $O(K^2)$ disks.
To summarize, we have added at most $O(K^2)$ disks in $\calH$.

\section{Time Complexity}
\label{sec:time}

The time complexity contains three parts. 
The first part is to enumerate all combination of $\calG$. 
We set $C = O(\frac{K^2}{\epsilon})$ (since we need $C>|\calH|/\epsilon$). 
Since $K = O(\frac{1}{\epsilon^2})$, the number of combinations is bounded by
$n^{C} =n^{O(1/\epsilon^5)}$. 
The second part is the construction of set $\calH$. 
It is easy to see that the time cost for each operation (i.e., label-cut) is no more than $O(n^2)$. Thus, the time cost is $O(K^2n^2) = O(n^2/\epsilon^4)$. 
The last part is the dynamic program. 
There are at most $O(K^2)$ substructures and at most $O(n^2)$ intersection points
in each substructure. 
Thus, the number of total states is at most $O((n^2)^{K^2})$. 
For each recursion, the time cost is at most $O(n)$. 
Thus, the overall time complexity of the dynamic program is
$O(n^{2K^2+1}) = n^{O(1/\epsilon^4)}$. 

Overall, the total time cost is 
$n^{O(1/\epsilon^5)}\cdot \max\{n^2/\epsilon^4,n^{O(1/\epsilon^4)} \}) =
  n^{O(1/\epsilon^9)}$.  
This finishes the proof of Lemma~\ref{thm:02}.

\section{Applications}
\label{sec:app}
The weighted dominating set problem (\MWDS) in unit disk graphs has numerous applications in the areas of wireless sensor networks~\cite{du2012connected}.
In this section, we show that our PTAS for \UDC\ can be used to obtain
PTASs for two important problems in this domain.

\subsection{Connected Dominating Set in UDG}
The goal for the \emph{minimum-weighted connected  dominating set} problem (\MWCDS)
is to find a dominating set which induces a connected subgraph and has the minimum
total weight.
Clark et al.~\cite{Clark91} proved that \MWCDS\  in unit disk graphs is
NP-hard.
Amb\"{u}hl et al.~\cite{Ambuhl06} obtained the first constant factor approximation algorithm for \MWCDS (the constant is 94).
The ratio was subsequently improved in a series of
papers~\cite{Huang09,dai20095,Erlebach10}.
The best ratio known is 7.105~\cite[pp.78]{du2012connected}.

One way to compute an approximation solution for \MWCDS\ is to first compute
minimum weighted dominating set (\MWDS) and then connect the dominating set
using a \emph{node-weighted
steiner tree} (\NWST) \cite{Huang09, Zou2009optimization}.
The optimal \MWDS\ value is no more than the optimal \MWCDS\ value.
After zeroing out the weight of all terminals, 
the optimal \NWST\ value (for any set of terminals) is also no more than the optimal \MWCDS\ value.
Hence, if there is an $\alpha$-approximation for \MWDS (or equivalently \UDC)
and a $\beta$-approximation for \NWST, then there is an
$\alpha+\beta$ factor approximation algorithm for \MWCDS.

Zou et al.~\cite{Zou2009optimization} show that there exists a $2.5\rho$-approximation for \NWST\
if there exists  a
$\rho$-approximation for the classical edge-weighted \emph{minimum steiner tree
problem}. The current best ratio for minimum steiner
tree is $1.39$~\cite{byrka2010improved}.
Thus, there exists a $3.475$-approximation for \NWST.
Combining with our PTAS for
\UDC, we obtain the following improved result for \MWCDS.
\begin{theorem}
	\label{thm:cds}
  There exists a polynomial-time $(4.475+\epsilon)$-approximation for \MWCDS\  for any fixed constant $\epsilon>0$.
\end{theorem}

\subsection{Maximum Lifetime Coverage in UDG}
The {\em maximum lifetime coverage problem} (\MLC) is a classical problem in wireless
sensor networks: Given $n$ targets $t_1 , \ldots, t_n$ and $m$ sensors $s_1, ...s_m$, each covering a subset of targets, find a family of sensor cover $S_1, \ldots, S_p$ with time lengths $\tau_1, \ldots, \tau_p$ in $[0,1]$, respectively, to maximize $\tau_1 + \ldots + \tau_p$ subject to that the total active time of every sensor is at most $1$.
\MLC\ is known to be NP-hard~\cite{Cardei05}.
Berman et al.~\cite{Berman05} reduced \MLC\  to the \emph{minimum weight sensor cover} (\MSC) problem through
Garg-K\"{o}nemann technique~\cite{Garg07}.
In particular, they proved that
if \MSC\ has a $\rho$-approximation,
then \MLC\ has a $(1+\epsilon)\rho$-approximation for any $\epsilon >
0$. Ding et al.~\cite{ding2012constant} noted that,
if all sensors and targets lie in the Euclidean plane and all sensors have the same covering radius, any approximation result for \UDC\ can be converted to almost the same approximation
for \MLC. Hence, the current best known result for \MLC\ is a $3.63$-approximation.
Using our PTAS, we obtained the first PTAS for \MLC.

\begin{theorem}
	\label{thm:mlc}
  There exists a PTAS for \MLC\  when  all sensors and targets lie in the Euclidean plane and all sensors have the same covering radius.
\end{theorem}

Let us mention one more variant of \MLC,
called maximum lifetime connected coverage problem,  studied by Du et al.\cite{DuHongwei13}.
The problem setting is the same as \MLC, except that
each sensor cover $S_i$ should induce a connected subgraph.
They obtained a $(7.105+\epsilon)$-approximation when
the communication ratios $R_c$ is no less two times
the sensing radius $R_s$.
Essentially, they showed that
an $\alpha$-approximation for \UDC\
and a $\beta$-approximation for \NWST\ imply
an $\alpha+\beta$-approximation algorithm for the connected \MLC\ problem.
Using our PTAS, we can improve the approximation ratio to $(4.475+\epsilon)$.


\section{Final Remarks}

Much of the technicality comes from the fact that the substructures interact 
each other in a complicated way
and it is not easy to ensure a globally consistent order.
The reader may wonder what if we choose more than two disks (but still a constant) in a small square,
hoping that the uncovered regions become separated and more manageable.
We have tried several other ways, like choosing a constant number of disks in the convex hull of the centers in
a small square. However, these seemed to only complicate, not to simplify, the matter.

We believe our result and insight are useful to tackle other problems 
involving unit disks or unit disk graphs.
On the other hand, our approach strongly relies on the special properties of
unit disks and does not seem to generalize to arbitrary disks with disparate radius.
Obtaining a PTAS for the weighted disk cover problem with arbitrary disks is still a central open problem in this domain.
An interesting intermediate step would be to consider the special case where the ratio between the longest radius and the shortest radius is bounded.

\bibliographystyle{acm}
\bibliography{WUDS}

\newpage
\appendix
\section{Missing Proofs}
\label{appendix}

\subsection{Missing Proofs in Section~\ref{sec:substructure}}
\label{apd:substructure}

\begin{lemma}
  \label{lm:pi}
  The central angle of any uncovered arc is less than $\pi$.
\end{lemma}

\begin{proof}
 We only need prove the arc of a gadget is less than $\pi$.
 As we union all gadgets and add more and more disks in $\calH$, the central angle of an arc only
 becomes smaller.
 So, now we fix a square gadget $\gadget(\ssquare)$.
 Consider the substructure above line $\dcenter_s\dcenter_t$.
 See the right hand side of Figure~\ref{fig:gadget} for an example.
 Suppose an arc $a$ with endpoints  $M$ and $M'$ on the boundary of $\gadget$.
 The center of the arc is in the central area of $\gadget$.
 If the central angle $\angle(a)\geq\pi$, its center should lie in the cap region bounded
 by $a$ and the chord $MM'$.
 Without loss of  generality, we suppose $M$ is closer to the line $\dcenter_s\dcenter_t$ than $M'$. Draw a auxiliary line at $M$ which is parallel to $\dcenter_s\dcenter_t$.
 If the line does not intersect  the central area, the center $a$ locates
 below the line $MM'$ (Hence, outside the cap region),
 thus the central angle is less than $\pi$.
 If the auxiliary line intersects the boundary of central area at point $N$.
 Then, we can see that
 $$
 |MN|=\sqrt{|MD_s|^2 - x^2}+|\dcenter_s\dcenter_t| - \sqrt{|ND_t|^2- x^2},
 $$
 where
 $x$ is the vertical distance between point $M$ and line
 $\dcenter_s\dcenter_t$.
 It is not difficult to see that $|MN| > 1$.
 It means that the
 center of the arc locates below the line $MN$ (Otherwise,
 the distance from the center to $M$ is larger than $|NM|$, which is larger than $1$,
 rendering a  contradiction).
 So, the central angle of any arc
 is  less than $\pi$.
\end{proof}

\subsection{Missing Proofs in Section~\ref{subsec:merge}}
\label{apd:merge}

{\bf Lemma~\ref{double-mix}}
{\em
Consider two non-empty small squares $\ssquare$, $\ssquare'$.
Suppose the square gadgets in the two
squares are $\gadget(\ssquare)=(\Disk_s, \Disk_t)$
and $\gadget'(\ssquare')=(\Disk_s', \Disk_t')$.
The active region pairs
$(\activeR_1, \activeR_2)$ and $( \activeR_1', \activeR_2' )$ are associated with
gadget $\gadget(\ssquare)$
and $\gadget'(\ssquare')$ respectively.
$(\activeR_1, \activeR_2)$ and $(\activeR_1',
\activeR_2')$ form a double-mixture.
Then, the following statements hold:
\begin{enumerate}
\item Their corresponding squares $\ssquare, \ssquare'$ are adjacent;
\item The core-central areas of $\gadget(\ssquare)$ and $\gadget'(\ssquare')$ overlap;
\item The angle between $\dcenter_s\dcenter_t$ and $\dcenter_s'\dcenter_t'$ is $O(\epsilon)$
\item None of the two core-central areas can overlap
with any small squares other than $\ssquare$ and $\ssquare'$.
\end{enumerate}
}

\begin{proof}

Suppose the core-central areas of $\gadget$ and $\gadget'$ are
$\coreCenter$ and $\coreCenter'$ respectively. Because  $( \activeR_1, \activeR_2)$ and
$(\activeR_1', \activeR_2') $ form a double-mixture.
There exists at
least one disk centered in $\ssquare'$ which appears in both $\activeR_1$ and $\activeR_2$.
Thus, the disk is centered in $\coreCenter$.
It means that $\coreCenter$ and $\coreCenter'$ overlap. Hence,
$\ssquare$ and $\ssquare'$ are adjacent. See Figure~\ref{fig:merge}.

 It is easy to see that any core-central area can overlap with at most two squares. Since
 $\coreCenter$ and $\coreCenter'$  overlap, at least one of them overlaps with both
 $\ssquare$ and $\ssquare'$. Without loss of generality,  suppose  $\coreCenter$ overlap with both
 $\ssquare$ and $\ssquare'$. We only need to prove $\coreCenter'$ cannot overlap with
 other squares other than $\ssquare$ and $\ssquare'$.

 Our proof needs a useful notion, called \emph{apex angle}.  Consider
 a square gadget. See the left one of Figure~\ref{fig:gadget} (and we use the notations there).
 The line $D_sE$ and
 $D_sF$ are the tangent lines to the boundary of core-central area at point $D_s$.
 We define
 \emph{apex angle} $\apex$  to be $\angle E\dcenter_sF$.
 We notice that $\apex$ only depends on the distance
  $|\dcenter_s\dcenter_t|$. When the side length of a small square is $\epsilon$,
  $\apex$ is $O(\epsilon)$.
  In fact,  even the core-central area is completely determined by $\dcenter_s$ and
  $\dcenter_t$.
  Thus, we can generalize the concept \emph{core-central area} by using any two overlapping
  disks in $\calH$.

Now go back to our proof.
See Figure~\ref{fig:merge}.
The segments $\dcenter_{s}\dcenter_{t}$ and $\dcenter_{s}'\dcenter_{t}'$ cannot intersect because they belong to
different small squares. Suppose the squares $\ssquare$ and $\ssquare'$ are
$A_1A_2A_5A_6$ and $A_2A_3A_4A_5$, respectively.
$A_2A_5$ is the common side.
W.l.o.g, suppose $\dcenter_s'$ is closer to line $A_2A_5$ than $\dcenter_t'$.
If any disk
centered in $\coreCenter'$ can appear in  $\activeR_1$ but outside the disk $\Disk_{s}'$
(i.e., outside the disks $\Disk_{s} \cup \Disk_{t} \cup \Disk_{s}'$),
the core-central areas defined by
$\dcenter_{s}'\dcenter_{t}$
	\footnote{$\dcenter_{s}'$ and $\dcenter_{t}$ do not locate in
	the same square. This is core-central area in the generalized sense..
	}
and $\coreCenter'$ should overlap nontrivial (not only touch at
point $\dcenter_{s}'$).
Thus, the angle $\angle\dcenter_{t}\dcenter_{s}'\dcenter_{t}'$ should
less than $O(\epsilon)$. It means  angle between $\dcenter_s\dcenter_t$ and $\dcenter_s'\dcenter_t'$
is at most $O(\epsilon)$. Moreover, angle between line $\dcenter_{s}'\dcenter_{t}'$
and $A_5A_4$ is no
less than  $\frac{\pi}{2} - O(\epsilon) > O(\epsilon)$. Thus, $\coreCenter'$ cannot intersect
$A_5A_4$. Hence, $\coreCenter'$ cannot overlap other squares.
\end{proof}

\begin{figure}[t]
  \centering
  \includegraphics[width=0.4\textwidth]{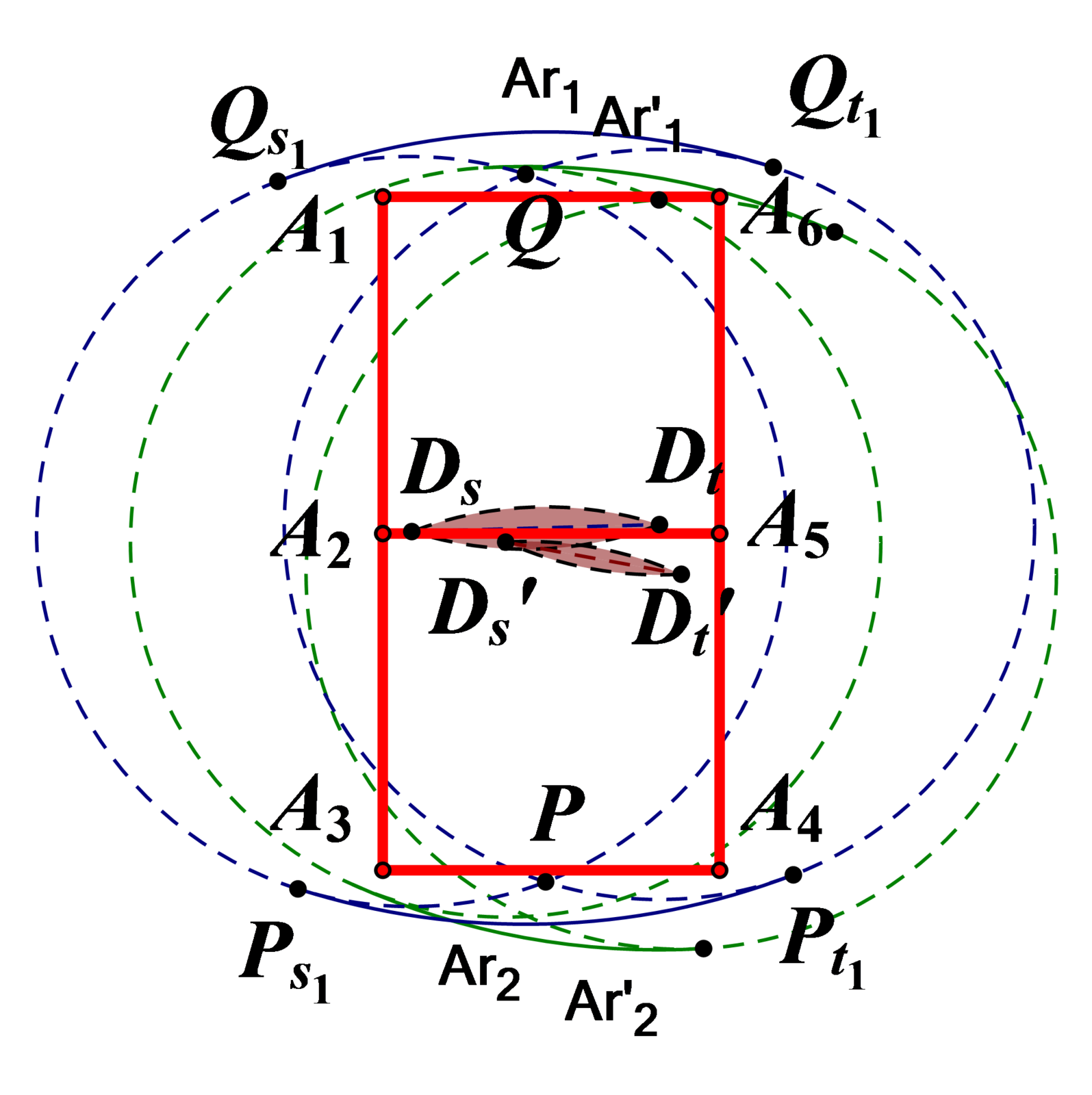}
  \caption{The case that arcs in two different active regions are not order-separable.
  	The two adjacent squares are $\ssquare=A_1A_2A_5A_6$ and $\ssquare'=A_2A_3A_4A_5$.
  	$(\Disk_s, \Disk_t)$ is the square gadget in $\ssquare$
  	and $(\Disk_s', \Disk_t')$ is the square gadget in $\ssquare'$.
  	The active
    regions $\activeR_1$ and $\activeR_2$ belong to the $\gadget(\ssquare)$,
    while the $\activeR_1'$ and $\activeR_2'$ belong to the $\gadget'(\ssquare')$.
    }
  \label{fig:merge}
\end{figure}

\noindent
{\bf Lemma~\ref{lm:no single mixture}}
{\em
Suppose active region pairs $(\activeR_1(\calA_1)$, $ \activeR_2(\calA_2) )$ and $(
\activeR_1'( \calA_1')$, $\activeR_2'(\calA_2'))$ are associated with gadget $\gadget$ and
$\gadget'$ respectively.
If $\calA_1$ and  $\calA_1'$ are
order-separable, then  $\calA_2$ and $\calA_2'$  are also order-separable.
}
\begin{proof}
We prove the lemma by contradiction. We only consider the arcs which have
siblings. Suppose $\calA_1$ and  $\calA_1'$ are order-separable but $\calA_2$
and $\calA_2'$  are not order separable. W.l.o.g., assume $a \prec a', \forall a
\in \calA_1, \forall a' \in \calA_1'$. Since $\calA_2$ and $\calA_2' $ are not
order-separable, there exist arcs  $b_1, b_2 \in \calA_2$ and $b' \in \calA_2'$
such that $ b_1 \prec b' \prec b_2$. (If not,  there exist arcs $b \in \calA_2$ and $b'_1,
b'_2 \in \calA_2'$ such that $b'_1 \prec b \prec b'_2$.)  Suppose $a_2$ is the
sibling of $b_2$, $a'$ is the sibling of $b'$. Thus, $a' \prec a_2$ based on the
same proof to Lemma~\ref{lemma:04:03}. It yields a contradiction to the assumption that $a \prec a',
\forall a \in \calA_1, \forall a' \in \calA_1'$.
\end{proof}

\noindent
{\bf Lemma~\ref{smallregion}}
{\em
  Consider the substructure $\substructure$ which contains an active region of square
  $\ssquare$.  The substructure can be  cut  into at most three smaller substructures, by
  doing label-cut twice.
  At most one of them
  contains the active region. Moreover,  this new substructure (if any) is
  bounded by the region $\Domain(\ssquare^+)$ associated with $\ssquare$.
}

\begin{proof}
See Figure~\ref{fig:smallActiveRegion}. We use the same notations defined 
in Section~\ref{subsec:merge}.
  The entire active region $[\bigcup_{i:\dcenter \in \coreCenter} \Disk_i - ( \Disk_s \cup
\Disk_t)] \cap H^{+}$ is a sub-region of  $(\Disk(P, 2) - \Disk_s - \Disk_t) \cap
H^{+}$. Thus,  any disk centered in $\Disk(Q_s, 1) \cap \Disk(Q_t,1) \cap H^{+}$ can
cover the active region. So, if there is a disk centered in $\Disk(Q_s, 1) \cap \Disk(Q_t,1) \cap
H^{+}$, we add it into $\calH$ such that the entire active region is covered.
If not,  we prove we can cut  $\substructure$ such that the new substructure contains the active
region is  bounded by the region   $(\Disk(D,1) - \Disk_s - \Disk_t )\cap H^{+}$.  The disks
centered in $\Region(\ssquare) \cap \Disk(Q_s,1) - \Disk(Q_s, 1) \cap \Disk(Q_t,1))$  cover
point $Q_s$. Thus, the arcs of these disks  are order-separable to the arcs in active
region. Similarly, the arcs of disks centered in $\Region(\ssquare) \cap \Disk(Q_t,1) - \Disk(Q_t,
1)\cap \Disk(Q_s,1)$ are also order-separable to the arcs in active region. Thus, we can add two
disks in $\calH$ to label-cut  $\substructure$ into
at most three new substructures. Obviously, there is only one of them containing arcs centered in
$\ssquare - \Disk(Q_s, 1) \cup \Disk(Q_t,1)$.
Hence, this new substructure is bounded in the region
$(\Disk(D,1) - \Disk_s - \Disk_t )\cap H^{+}$.
\end{proof}

\subsection{Missing proofs in Section~\ref{subsec:cutoff}}
\label{apd:cutoff}

{\bf Lemma~\ref{arcuniqueness}}
{\em
  In each of the above iterations, one substructure $\substructure(\baseline, \calA)$ is cut into
  at most two new substructures. Any original arc in
  $\calA$ cannot be cut into two pieces, each of which belongs to a different new substructure.
}

\begin{proof}
Figure~\ref{fig:point-cut-off} gives an explanation about the change of
substructures before and after the process.
Suppose the envelope is $\{a_1[Q_s, Q_1], \ldots, a_k[Q_{k-1},Q_t] \}$.
After we add the disk $\Disk(a_i)$ in
$\calH$, the baselines of the two new substructures are
$\baseline_1 = \baseline[Q_s, P_{i+1}] \cup
a_{i+1}[P_{i+1}, Q_i]$ and
$\baseline_2 = a_{i+1}[Q_{i+1},P_{i+1}'] \cup \baseline[P_{i+1}', Q_t]$.
Since $a_{i+1}$ is an arc lying on the envelope, there does not exist an arc $b$ with endpoints
$(A, B)$ such that $A \prec P_{i+1} \prec P_{i+1}' \prec B$.
So the endpoints of any arc in
$\substructure_1$ cannot be in $\Region(\substructure_2)$. Thus, any arc cannot be separated
into two substructures.
\end{proof}

\subsection{Missing proofs in Section~\ref{subsec:GOC}}
\label{apd:GOC}

\begin{lemma}
  \label{pi/4}
  Consider two disks $\Disk_u$ and $\Disk_l$ in squares $\ssquare_u, \ssquare_l$
respectively. Suppose $\Disk_u \notin \gadget(\ssquare_u) , \Disk_l \notin \gadget(\ssquare_l)$ and $\Disk_u$ and
$\Disk_l$ intersect at points $A$ and $B$. Suppose the side length of square is $\mu$. If
$\angle A\dcenter_uB$ (or $\angle A\dcenter_lB$) is more than $2\sqrt{2\mu}$, $\gadget(\ssquare_u)$ and
$\gadget(\ssquare_l)$ should overlap. Moreover, if intersection $A$ (resp. $B$) is in $\Uncover$, the two
arcs which intersect at point A (resp. $B$) are in the same substructure.
\end{lemma}

\begin{proof}
 Suppose $\Disk_u$ and $\Disk_l$ intersect at point $A$ and $B$. $\Disk_s \in \gadget(\ssquare_u)$, $\Disk_s' \in
 \gadget(\ssquare_l)$. Thus, $\Disk_u, \Disk_s$ locate in the same square $\ssquare_u$ and $\Disk_l, \Disk_s'$
 locate in the same square. The line $\dcenter_u\dcenter_l$  intersect disk $\Disk_s$  at point $E$ and
 intersect disk $\Disk_s'$ at point $F$. $|\dcenter_uE|$ and $|\dcenter_lF|$ are no less
than $1-\mu$ according to the triangle inequality. Suppose $\angle A\dcenter_uB$ is
$\theta$. Then the length
$|\dcenter_u\dcenter_l|$ equals $2 \cos(\frac{\theta}{2})$. Since $\theta > 2\sqrt{2\mu}$, we have
$|\dcenter_uE| + |\dcenter_lF| > |\dcenter_u\dcenter_l|$, i.e. $\Disk_s$ and $\Disk_s'$
overlap. Figure~\ref{fig: pi/4} illustrates the situation. If $A$ (or $B$) is in $\Uncover$, it
is obvious that the two arcs which intersect at $A$ (or $B$) can cover the same point on the baseline (i.e., the
intersection point of $\Disk_s$ and $\Disk_s'$) which means they are in the same substructure.
\end{proof}

\begin{figure}[h]
  \centering
  \includegraphics[width=0.3\textwidth]{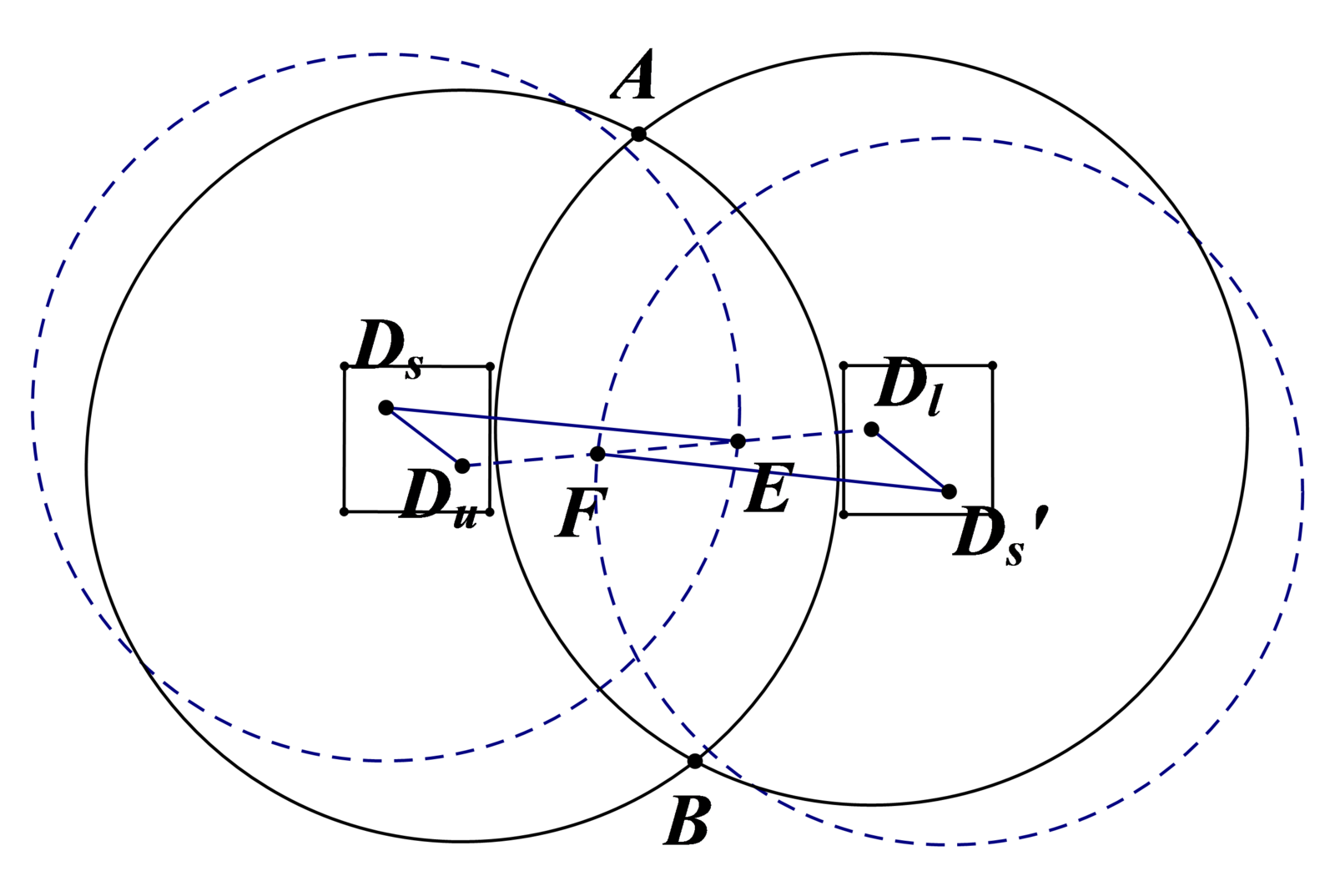}
  \caption{$\Disk_u$ and $\Disk_l$ intersect at point $A$ and  $B$. Suppose the side length of
    square is $\mu$. If the central angle of arc
    $a[A,B]$ is more than $2\sqrt{2\mu}$.   $\Disk_s, \Disk_s'$ should
    overlap, where $\Disk_s \in \gadget(\ssquare_u), \Disk_s' \in \gadget(\ssquare_l)$.}
  \label{fig: pi/4}
\end{figure}

Based on the Lemma~\ref{pi/4},  we prove  Lemma~\ref{oneByone} as follows.

\noindent
{\bf Lemma~\ref{oneByone}}
{\em
  We are given a substructure $\substructure(\baseline,\calA)$.
  $a \in \calA$ and two arcs $b_1, b_2 \notin \calA$.
  If $b_1, b_2$ cover the same point on
  $a$,  $b_1, b_2$ should belong to the same substructure.
}

\begin{proof}
  We prove the lemma by contradiction.  Suppose the arcs $a, b_1, b_2$ are parts of $\Disk, \Disk_1, \Disk_2$
respectively. Consider a point $P$ on $\Disk$ that can be covered by all of $a, b_1, b_2$.  $\Disk$ and
$\Disk_1$ intersect at $A_1$ and $B_1$, meanwhile $\Disk$ and $\Disk_2$ intersect at $A_2$ and
$B_2$. See Figure~\ref{fig:onebyone}. Since $\Disk$ and $\Disk_1$ belong to
different substructures, we know  $\angle A_1DB_1$ is less than $2\sqrt{2\mu}$ based on
Lemma~\ref{pi/4}. Similarly, $\angle A_2DB_2$ is less than $2\sqrt{2\mu}$. It is easy to see that
$\angle D_1DD_2$ is less than $2\sqrt{2\mu}$. When $\angle D_1DD_2$ is $O(\sqrt{\epsilon})$ when $\mu = O(\epsilon)$. $|D_1D_2|$ is less than $1$ as
$|DD_1|$ and $|DD_2|$ are more than $1$ and less than 2.
Suppose $A_3$ and $B_3$ are the two intersection points of $\Disk_1$ and $\Disk_2$.
Thus, we have $\angle A_3 D_1 B_3 > \pi/3> 2\sqrt{2\mu}$.
Moreover, $P$ is in $\Uncover$.
By Lemma~\ref{pi/4}, we can see that it contradicts the fact that $\Disk_1$ and $\Disk_2$ belong to
two different substructures. 
\end{proof}

\begin{figure}[h]
  \centering
  \includegraphics[width=0.3\textwidth]{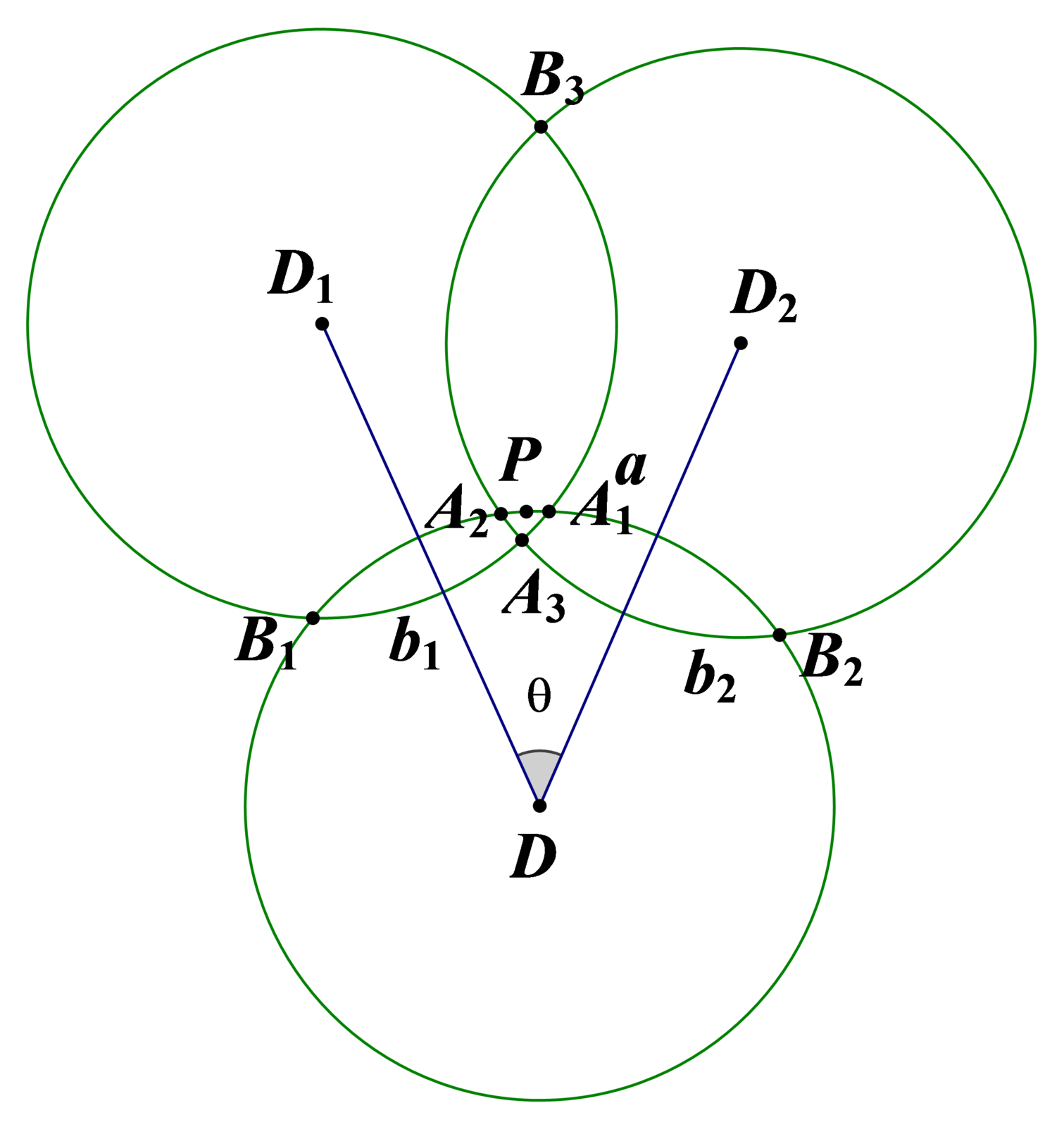}
  \caption{ Consider a point $P$ on $\Disk$. $a, b_1, b_2$ which belong to disks $\Disk, \Disk_1,
\Disk_2$ respectively, can cover the point $P$. $\Disk$ and $\Disk_1$ intersect at $A_1$ and $B_1$,
meanwhile $\Disk$ and $\Disk_2$ intersect at $A_2$ and $B_2$. Then, $\Disk_1$ and $\Disk_2$
belong to the same substructure.}
  \label{fig:onebyone}
\end{figure}

\end{document}